\documentclass[letterpaper, 10pt,twocolumn,twoside]{IEEEtran}
\makeatletter
\let\NAT@parse\undefined
\makeatother

\usepackage{graphics} 
\usepackage{epsfig} 
\usepackage{mathptmx} 
\usepackage{times} 
\usepackage{amsmath} 
\usepackage{amssymb}  
\usepackage{xcolor}  
\usepackage{subfigure}
\usepackage{caption} 
\usepackage{amsthm}

\usepackage[hyphens]{url} %
\usepackage[numbers]{natbib}

\newcommand{\squeeze}{-1.5mm}
\newcommand{\squeezeSection}{-5mm}

\newtheorem{definition}{Definition}
\newtheorem{assumption}{Assumption}
\newtheorem{proposition}{Proposition}

\newtheorem{corollary}{Corollary}

\makeatletter
\def\subsubsection{\@startsection{subsubsection}
                                 {3}
                                 {\z@}
                                 {0ex plus 0.1ex minus 0.1ex}
                                 {0ex}
                                 {\normalfont\normalsize\itshape}}
\makeatother
                                 
\title{\LARGE \bf Integrated Analysis of Coarse-Grained Guidance for Traffic Flow Stability
}

\author{Sirui Li, Roy Dong, Cathy Wu
\thanks{Sirui Li is with the Institute for Data, Systems, and Society, Massachusetts Institute of Technology,
        Cambridge, MA, 02139, USA.
        {\tt\small siruil@mit.edu}}%
\thanks{Roy Dong is with the Industrial \& Enterprise Engineering department at the University of Illinois at Urbana-Champaign, Urbana, IL, 61801, USA.
{\tt\small roydong@illinois.edu}}
\thanks{Cathy Wu is with the Laboratory for Information \& Decision Systems; the Institute for Data, Systems, and Society; and the Department of Civil and Environmental Engineering, Massachusetts Institute of Technology,
        Cambridge, MA, 02139, USA.
        {\tt\small cathywu@mit.edu}}%
}

\begin{document}

\maketitle
\thispagestyle{empty}
\pagestyle{empty}

\begin{abstract}
Autonomous vehicles (AVs) enable more efficient and sustainable transportation systems. Ample studies have shown that controlling a small fraction of AVs can smooth traffic flow and mitigate traffic congestion. However, deploying AVs in real-world systems poses challenges due to safety and cost concerns. A viable alternative approach that can be implemented in the near future is \textit{coarse-grained guidance}, where human drivers are guided by real-time instructions, updated every $\Delta$ seconds, to stabilize the traffic. While previous theoretical studies consider stability analysis for continuous AV control, this article presents the first integrated theoretical analysis that directly relates the guidance provided to the human drivers to the traffic flow stability outcome. Casting the problem into the Lyapunov stability framework, this study derives sufficient conditions for coarse-grained guidance with hold length $\Delta$ to stabilize the system, and provides extensions of the analysis to incorporate additional human driving behaviors such as magnitude error and reaction delay. Numerical simulations reveal that the theoretical analysis closely matches simulated results. The analysis further offers insights into the relationship between system parameter and stability criteria, and can be leveraged to design improved controllers with greater maximum hold length.
\end{abstract}

\section{INTRODUCTION}
\label{sec:introduction}
Transportation remains a key priority for public health and safety, climate change mitigation, and economic competitiveness~\cite{national2019traffic, epaghgemission, winston2015transportation}. Everyday traffic is deeply intertwined with these priorities. For instance, studies suggest that mitigating congestion could reduce up to 20\% of CO$_2$ emissions~\cite{barth2008real}.

\begin{figure}
    \centering
    \includegraphics[width=0.5\textwidth]{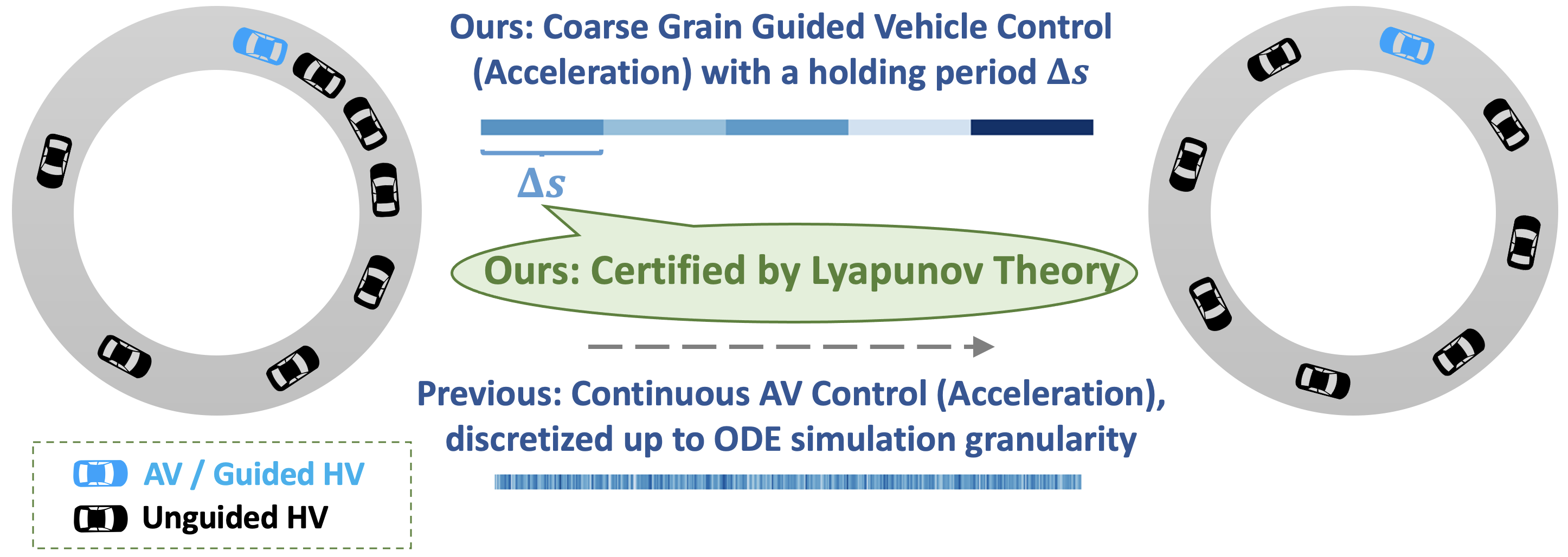}
    \caption{\textbf{Overview.} Consider a traffic system with an AV or guided human-driven vehicle (in blue) and unguided human-driven vehicles (in black). Previous work continuously updates the control action of the AV. In contrast, this work investigates coarse-grained guidance, which periodically (every $\Delta$ seconds) updates the control action. Theoretical guarantees on the hold limit are provided using Lyapunov analysis.}
    \vspace*{-0.3cm}
    \label{fig:illustration}
\end{figure}

Autonomous vehicles (AVs) are a long-anticipated solution to congestion. A single AV has been demonstrated to dampen stop-and-go waves in a circular track with 22 vehicles in a field experiment~\cite{stern2018dissipation}. Up to a 57\% improvement in average velocity is possible in simulation, when using an AV controller designed with deep reinforcement learning~\cite{wu2021flow}. However, it remains challenging to deploy AVs due to safety concerns. 

Providing real-time guidance to human drivers can serve as a middle ground solution to deploy traffic stabilizing control at scale in the near future. This article is motivated by numerous field-tested studies that demonstrate the efficacy of providing real-time guidance to drivers as a means to alter traffic flow~\cite{munoz2013validating, mintsis2017evaluation, bhadani2018dissipation, nice2021can}. In other words, \textit{coarse-grained guidance}, as opposed to fine-grained guidance that requires AVs, has the potential to unlock large societal benefits in the near term.

However, due to the cost of field tests, it is impractical to provide coverage of all or even common field conditions.
A theoretical foundation is thus crucial, one that can provide assurances for leveraging human drivers to alter traffic flow. This article takes a step in this direction by formalizing aspects of coarse-grained guidance as \textit{piecewise-constant control}---also known as zero-order hold (ZOH) control---and theoretically analyzing its implications for traffic flow stability, as depicted in Fig.~\ref{fig:illustration}. While previous theoretical studies consider stability analysis for continuous AV control~\cite{cui2017stabilizing, zheng2020smoothing}, this work is the first to consider an integrated analysis that directly relates the real-time guidance provided to the human drivers to the system-level stability outcome. The theory accounts for the interaction between a single guided human-driven vehicle under coarse-grained guidance (guided HV) and the rest of the unguided human-driven vehicles (unguided HVs), governed by the Optimal Velocity Model (OVM)~\cite{bando1995dynamical}. Lyapunov functions and Lyapunov-Krasovskii functionals are leveraged to provide sufficient stability conditions for the traffic system under the piecewise-constant control with a hold length $\Delta$. Through numerical analysis, we demonstrate that the theoretical conditions closely match empirical simulations under a variety of OVM parameters, and hence the theory serves as a reliable certificate of the \textit{hold limit} for the system, which we define as the maximum hold length that guarantees system's stability. 

This article extends and subsumes~\cite{li2023ecc}. Our overall contributions are:
\begin{itemize}
    \item We cast coarse-grained guidance with piecewise-constant control into the sample-data system framework, and propose a theoretical framework with Lyapunov analyses to provide sufficient conditions for stabilizing the traffic system with a single piecewise-constant controlled vehicle. 
    \item We perform extensive numerical analysis to show that the theory closely match empirical simulations. Notably, the Lyapunov-Krasovskii functionals closely match the empirical hold limits in both trend and absolute value.
    \item Beyond~\cite{li2023ecc}, we (1) derive detailed insights into the relationship between OVM parameters and stability criteria, (2) design piecewise-constant controllers with longer hold limits, (3) discuss applications of the analyses to a broader class of coarse-grained guidance such as piecewise-constant velocity guidance in addition to acceleration guidance, and (4) expand the analyses to a broader class of human-compatible driving by incorporating magnitude error and human reaction delay.
\end{itemize}

\section{RELATED WORK}
\subsection{Traffic stabilization with autonomous vehicles}
There has been growing interest in controlling autonomous vehicles to stabilize mixed traffic systems of autonomous and human vehicles. A few works use reinforcement learning to design controls in various scenarios such as stabilizing stop-and-go waves in the low AV-adoption regime~\cite{wu2021flow}, coordinating AVs to exhibit traffic light behaviors~\cite{yan2022unified}, and designing eco-driving Lagrangian controls to reduce fuel consumption~\cite{guo2021hybrid}. Theoretical studies have been carried out on the linearized continuous system of the Optimal Velocity Model (OVM)~\cite{bando1995dynamical} for the ring-road traffic setting, under two main stability concepts 1) asymptotic stability~\cite{cui2017stabilizing, zheng2020smoothing, zheng2015stability, zhu2018analysis, wang2020controllability, mousavi2022synthesis}, and 2) string stability~\cite{swaroop1994string, bose2003analysis, rogge2008vehicle, giammarino2020traffic, liu2022structural}. Linear (asymptotic) stability analysis under disturbance, uncertainty, and reaction delay have also been proposed~\cite{mousavi2022synthesis, jin2016optimal}. While most of the analytic studies rely on linearizing the system,~\citet{gisolo2022nonlinear} introduces a stability analysis based on sector nonlinearity.
Our work follows the asymptotic stability concept, applying Lyapunov analysis to the linearized system. The aforementioned literature also derives continuous optimal controllers, with numerical simulations to demonstrate their effectiveness in stabilizing traffic flow. This body of work is foundational for our analysis. However, it is not directly applicable due to the continuous updates to the control action of the vehicles.

\vspace{\squeezeSection}
\subsection{Guiding human drivers to alter traffic flow}
The basic model we analyze most closely follows the simulation-based studies of~\citet{sridhar2021piecewise, sridhar2021learning}, which propose the class of piecewise-constant driving policies for guiding drivers to mitigate congestion by providing periodic instructions every $\Delta$ seconds. The study leverages deep reinforcement learning to synthesize controllers that are effective in stabilizing traffic flow in the ring road, including in the presence of lane changes. While a theoretical analysis is provided by~\citet{sridhar2021piecewise}, it does not consider interactions between the guided HV with other unguided HVs. In contrast to~\citet{sridhar2021piecewise}'s focus on average case velocity, we focus on worst case stability. Hence, the two works are not directly comparable. The piecewise-constant driving policies belong to a broader class of driving guidance that uses simple and easy-to-follow interventions to achieve desired traffic outcomes as discussed in Sec.~\ref{sec:introduction}. 

\subsection{Sample-data systems and Lyapunov stability analysis}
The class of piecewise-constant policies is conceptually similar to the zero-order hold sample-data systems~\cite{chen2012optimal}, where a continuous system is controlled by a digital holding device. The device takes a digital input every $\Delta$ seconds to produce a digital control being held constant for the entire holding period of length $\Delta$. In contrast to such systems, which typically are designed for hold lengths of milliseconds or less, we consider longer hold lengths of tens of seconds to respect human reaction times.

Lyapunov functions have been used to analyze general control systems with discontinuous feedback~\cite{clarke2010discontinuous}, of which our piecewise-constant coarse-grained guidance is a special case. To incorporate delays in human driver reaction times, Lyapunov-Krasovskii functionals have been used~\cite{li2014stabilization}, albeit in the context of human drivers issuing continuous controls based on delayed input states, and hence, the controlled system is still continuous. In contrast, our work considers piecewise-constant controls that are updated every $\Delta$ seconds, and hence belongs to the sample-data system paradigm. Prior studies~\cite{fridman2014tutorial, fridman2001new, liu2012wirtinger} adopt Lyapunov-Krasovskii functionals to general sample-data systems, and show that tailored Lyapunov-Krasovskii functionals outperforms general time-delay Lyapunov-Krasovskii functionals on toy sample-data control examples. Our work is the first to apply a sample-data Lyapunov-Krasovskii functional to analyze system-level stability of coarse-grained guidance, and validate through simulation that the theoretical guarantees closely align with simulated results.

\section{PRELIMINARIES.}
Following~\citet{zheng2020smoothing}, we consider a single-lane ring road with circumference $L$ and $n$ vehicles (see Fig.~\ref{fig:illustration}). Let the position of $i$-th vehicle be $p_i(t)$, the velocity be $v_i(t) = \dot{p}_i(t)$, the spacing be $s_i(t) = p_{i-1}(t) - p_i(t)$, and the acceleration be $a_i(t) = \dot{v}_i(t)$.

The standard car following model (CFM) for the unguided HVs takes the nonlinear form
\vspace{\squeeze}
\begin{equation}
    \dot{v}_i(t) = F(s_i(t), \dot{s}_i(t), v_i(t)),
\vspace{\squeeze}
\end{equation}
where the uniform flow equilibrium achieved at spacing $s^*$ and velocity $v^*$ such that $F(s^*, 0, v^*) = 0$.

Denote the error state as $\tilde{s}_i(t) = s_i(t) - s^*$ and $\tilde{v}_i(t) = v_i(t) - v^*$, the linearization of the CFM around the equilibrium is 
\begin{equation}
    \begin{cases}
    \dot{\tilde{s}}_i(t) & = \tilde{v}_{i-1}(t) - \tilde{v}_i(t)\\
    \dot{\tilde{v}}_i(t) & = a_1 \tilde{s}_i(t) - a_2 \tilde{v}_i(t) + a_3 \tilde{v}_{i-1}(t)
    \end{cases}
    \label{eq:hv_ode}
\end{equation}
where $a_1 = \frac{\partial F}{\partial s}, a_2 = \frac{\partial F}{\partial \dot{s}} - \frac{\partial F}{\partial v}, a_3 = \frac{\partial F}{\partial \dot{s}}$ evaluated at $(s^*, v^*)$.

The Optimal Velocity Model (OVM)~\cite{bando1995dynamical} follows the form
\begin{equation}
    F(s_i(t), \dot{s}_i(t), v_i(t)) = \alpha(V(s_i(t)) - v_i(t)) + \beta \dot{s}_i(t),
    \label{eq:ovm}
\end{equation}
where $\alpha > 0, \beta > 0$, and the optimal velocity $V(s_i(t))$ is
\vspace{\squeeze}
\begin{equation}
    V(s) = \begin{cases}
    0, & s \leq s_{st}\\
    f_v(s), & s_{st} < s < s_{go}\\
    v_{max}, & s \geq s_{go}
    \end{cases}    ,
    \label{eq:vclipping}
\vspace{\squeeze}
\end{equation}
and $f_v(s)$ typically takes the form
\vspace{\squeeze}
\begin{equation}
    f_v(s) = \frac{v_{max}}{2}\Big(1 - \cos \big(\pi \frac{s - s_{st}}{s_{go} - s_{st}}\big)\Big).
    \label{eq:ovm_f}
\vspace{\squeeze}
\end{equation}
As a result, $v^* = V(s^*), a_1 = \alpha \dot{V}(s^*), a_2 = \alpha + \beta, a_3 = \beta$.

\section{Coarse-grained Guidance Modeling}
We consider a system with one guided HV $i=1$ under piecewise-constant control with hold length $\Delta$, and $n-1$ unguided HVs under OVM. At a given time $t \in [t_k, t_{k+1}]$ where $[t_k, t_{k+1}]$ is the corresponding holding period, the CFM for the piecewise-constant controlled vehicle takes the form
\vspace{\squeeze}
\begin{equation}
    \dot{v}_1(t) = f\left(u\left(z(t_k)\right); s_1(t), \dot{s}_1(t), v_1(t)\right),
\vspace{\squeeze}
\end{equation}
where $z(t) = [s_1(t), v_1(t), ..., s_n(t), v_n(t)]$ denotes the state vector at any time $t$, the term $u(z(t_k))$ represents the piecewise-constant controller, and we allow a dynamics function $f$ which may additionally depend on the state of the guided HV $(s_1(t), \dot{s}_1(t), v_1(t))$. Examples of the CFM are provided next.

For a class of piecewise-constant velocity guidance, a constant desired velocity $u(z(t_k))$ is proposed to the guided HV during each holding period; the vehicle uses an OVM-like dynamics $f$ to reach the desired velocity, resulting in
\vspace{\squeeze}
\begin{equation}
    \dot{v}_1(t) = \alpha(u(z(t_k)) - v_1(t)) + \beta \dot{s}_1(t)
    \label{eq:velocity_guidance}.
\vspace{\squeeze}
\end{equation}
Meanwhile, for a class of piecewise-constant acceleration guidance, a constant acceleration is imposed during the holding period ($f$ is the identity function), resulting in
\vspace{\squeeze}
\begin{equation}
    \dot{v}_1(t) = u(z(t_k)).
    \label{eq:acceleration_control}
\vspace{\squeeze}
\end{equation}
We follow previous works~\cite{sridhar2021piecewise, sridhar2021learning} to focus on the piecewise-constant acceleration guidance in this work. 

Considering a full state feedback piecewise constant control $u(z(t_k)) = K z(t_k)$. Lumping the error state into a vector form with $x(t) = [\tilde{s}_1(t), \tilde{v}_1(t), ..., \tilde{s}_n(t), \tilde{v}_n(t)]^\intercal = z(t) - x^*$ where $x^* = [s^*, v^*, ..., s^*, v^*]$ is the equilibrium state, the error dynamics for the controlled vehicle is given by
\vspace{\squeeze}
\begin{equation}
    \begin{cases}
    \dot{\tilde{s}}_1(t)& = \tilde{v}_n(t) - \tilde{v}_1(t)\\
    \dot{\tilde{v}}_1(t)& = K x(t_k)
    \end{cases}    ,
    \label{eq:av_ode}
\vspace{\squeeze}
\end{equation}
and the error dynamics of the linearized piecewise-constant control system is thus given by
\vspace{\squeeze}
\begin{equation}
    \dot{x}(t) = A x(t) + A_1 x(t_k), k = 0, 1, ...
    \label{eq:matrix_system}
\vspace{\squeeze}
\end{equation}
with 
\vspace{\squeeze}
\begin{equation}
    A = \begin{bmatrix} 
        	C_1 & 0 & ... &  ... & 0 & C_2 \\
        	D_2 & D_1 & 0 & ... & ... & 0\\
        	0 & D_2 & D_1 & 0 & ... & 0 \\
        	\vdots & \ddots & \ddots & \ddots & \ddots & \vdots\\
        	0 & ... & 0 & D_2 & D_1 & 0 \\
        	0 & ... & ... & 0 & D_2 & D_1
    	\end{bmatrix},
    B = \begin{bmatrix} 
        	B_1 \\
        	B_2 \\ 
        	B_2 \\
        	\vdots \\
        	B_2
    	\end{bmatrix},
\end{equation}
\begin{equation}
\begin{aligned}
    \text{where} \quad \quad D_1 = \begin{bmatrix}
            0 & -1 \\
            a_1 & - a_2
          \end{bmatrix} &,
    D_2 = \begin{bmatrix}
            0 & 1 \\
            0 & a_3
          \end{bmatrix}, \\[5pt]
    C_1 = \begin{bmatrix}
            0 & -1 \\
            0 & 0
          \end{bmatrix},
    C_2 = \begin{bmatrix}
            0 & 1 \\
            0 & 0
          \end{bmatrix} & ,
    B_1 = \begin{bmatrix}
            0 \\
            1
          \end{bmatrix}, 
    B_2 = \begin{bmatrix}
            0 \\
            0
          \end{bmatrix}.
    \label{equ:notations}
\end{aligned}
\end{equation}
and $A_1 = - B K$ represents the full state feedback piecewise-constant control coefficients. 

The formulation directly extends~\citet{zheng2020smoothing} to consider piecewise-constant acceleration guidance. We remark that the formulation aligns with the sample-data system framework~\cite{chen2012optimal} with zero-order hold. We further note that different classes of piecewise-constant controls can be modeled similarly using different $A$ and $A_1$ matrices. For example, the velocity guidance (Eq.~(\ref{eq:velocity_guidance})) can be expressed with
\vspace{\squeeze}
\begin{equation}
    C_1 = \begin{bmatrix}
            0 & -1 \\
            0 & -a_2
          \end{bmatrix},
    C_2 = \begin{bmatrix}
            0 & 1 \\
            0 & a_3
          \end{bmatrix},
    B_1 = \begin{bmatrix}
            0 \\
            \alpha
          \end{bmatrix}, 
\vspace{\squeeze}
\end{equation}
and the remaining matrices unchanged. The matrices $C_1$ and $C_2$ follow the unguided HV representations $D_1$ and $D_2$, except the $a_1$ term representing the desired velocity is moved from $D_1$ in the uncontrolled system matrix to the $\alpha$ term in the control submatrix $B_1$. The Lyapunov analyses in Sec.~\ref{sec:lyapunov} naturally apply to the broader classes of piecewise-constant controls, as they are agnostic to the specific form of $A$ and $A_1$.  

\section{LYAPUNOV ANALYSIS}
\label{sec:lyapunov}
\subsection{A Lyapunov bound}
\label{sec:lyapunov1}
We first derive a lower bound on the hold limit using Lyapunov theory. While previous literature provide a Lyapunov bound on general nonlinear systems with discontinuous controls~\cite{clarke2010discontinuous}, we adapt and modify the derivation to general linearized systems where the controls are piecewise-constant. In later sections, we apply the bound to the specific linearized ring-road OVM to extract meaningful insights into the traffic system. Although our focus is on the ring-road, the following Lyapunov analyses are derived on general linearized system matrices $A$ and $A_1$, and thus can be readily applied to other traffic topologies, such as open-road (see~\cite{mousavi2022synthesis} for $A$ and $A_1$ specifications).
\begin{proposition} \label{prop:lyapunov}
Let there exist $n\times n$ matrices $P > 0, Q > 0$ such that $V(x) = x^\intercal P x > 0$ with $\dot{V}(x) = -x^\intercal Q x < 0$ and $-Q = (A + A_1)P + P(A + A_1)^\intercal$ is a valid Lyapunov function for the linear system with continuous full-state feedback control, $\dot{x}(t) = (A + A_1)x(t)$ where $A_1 = -BK$. Then the sample-data system with piecewise constant control (\ref{eq:matrix_system}) is asymp. stable for hold length 
\begin{equation}
    \Delta \leq c' \frac{\sigma_{\min}(Q)}{\sigma_{\max}(P) (\sigma_{\max}(A) + \sigma_{\max}(A_1))^2}, 
    \label{eq:lyapunov_theory}
\end{equation}
up to a scaling constant $c' > 0$, where $\sigma_{min}(\cdot)$ and $ \sigma_{max}(\cdot)$ are the minimum and maximum singular value of the corresponding matrix.\\
\end{proposition}
\vspace{-0.7cm}
\begin{proof} Consider a time period $[t_k, t_{k+1}]$ with $t_{k+1} - t_k \leq \Delta$. We use the Lyapunov function for the continuous system $V(x) = x^\intercal P x$, and show that it is a valid Lyapunov function for the sample-data system by showing $V(x(t)) - V(x(t_k))$ is sufficiently negative, i.e. $V(x(t))$ decreases as $t$ increases. We have for all $t \in [t_k, t_{k+1}]$:
\begin{equation}
\begin{aligned}
     & V(x(t)) - V(x(t_k)) \\
    = \; & \langle\nabla V(x(t^*)), \dot{x}(t^*)\rangle(t-t_k) \quad \text{for some }t^* \in (t_k, t)\\
    = \; & \langle\nabla V(x(t_k)), \dot{x}(t_k)\rangle(t-t_k) \\
    + \; &\langle\nabla V(x(t_k)), \dot{x}(t^*) - \dot{x}(t_k)\rangle(t-t_k)\\
    + \; & \langle\nabla V(x(t^*)) - \nabla V(x(t_k)), \dot{x}(t_k)\rangle(t-t_k),
    \label{eq:lyapunov_theory_block1}
\end{aligned}
\end{equation}
where the first equality holds by the mean value theorem. For the last equality,
the first term gives a decrease in Lyapunov value, as at time $t_k$ the system behaves the same as the continuous system with continuous control using the instantaneous state information $x(t_k)$. Specifically, $\langle\nabla V(x(t_k)), \dot{x}(t_k)\rangle  = x(t_k) ((A + A_1)P + P(A + A_1)^\intercal) x(t_k) = - x(t_k) Q x(t_k) \leq  -\sigma_{\min}(Q) \|x(t_k)\|_2^2 \; \leq \; 0$. The last two terms represent the perturbation incurred by the piecewise-constant control, where $\nabla V(x(t_k)) = 2 x(t_k)^\intercal P$, $\dot{x}(t^*) - \dot{x}(t_k) = A(x(t^*) - x(t_k))$, $\nabla V(x(t^*)) - \nabla V(x(t_k)) = 2(x(t^*) - x(t_k))^\intercal P$, and $\dot{x}(t_k) = (A + A_1)x(t_k)$. The following worst-case bounds hold:
\begin{equation*}
    \begin{aligned}
 \left\| \nabla V(x(t_k))\right\|_2 & \leq 2 \sigma_{max}(P) \left\|x(t_k)\right\|_2,\\
   \left\|\dot{x}(t_k)\right\|_2 & \leq \sigma_{\max}(A + A_1) \left\|x(t_k)\right\|_2,\\
    \left\|\nabla V(x(t^*)) - \nabla V(x(t_k))\right\|_2 &\leq 2\sigma_{max}(P) \left\|x(t^*) - x(t_k)\right\|_2,\\
\end{aligned}
\end{equation*}
\begin{equation}
\begin{aligned}
   \left\|x(t^*) - x(t_k)\right\|_2 & = \left\|\int_{t_k}^{t^*} \dot{x}(s) ds\right\|_2
   \leq (t^* - t_k) \max\limits_{s \in [t_k, t^*]} \left\|\dot{x}(s)\right\|_2 \\
   & \leq \Delta  \max\limits_{s \in [t_k, t^*]} \left\|A x(s) + A_1 x(t_k)\right\|_2\\
   & \leq \Delta (\sigma_{\max}(A) + \sigma_{\max}(A_1)) \max\limits_{s \in [t_k, t_{k+1}]} \left\|x(s)\right\|_2 .
   \label{eq:lyapunov_theory_block4}
\end{aligned}
\end{equation}
Taken together, we have
\begin{equation}
\begin{aligned}
    & \quad V(x(t)) - V(x(t_k)) \\
    & \leq (t-t_k) \big(- \sigma_{\min}(Q) \|x(t_k)\|_2^2 \\
    & \quad + 2 \sigma_{\max}(P)\|x(t_k)\|_2 \sigma_{\max}(A) \|x(t^*) - x(t_k)\|_2 \\
    & \quad + 2 \sigma_{\max}(P) \|x(t^*) - x(t_k)\|_2 \sigma_{\max}(A + A_1) \|x(t_k)\|_2\big) \\
    & = (t-t_k) \big(- \sigma_{\min}(Q) \|x(t_k)\|_2^2   + 2 \sigma_{\max}(P) \big(\sigma_{\max}(A) \\
    & \quad \quad \quad \quad \quad \quad \;\; + \sigma_{\max}(A + A_1)\big) \times \|x(t_k)\|_2 \|x(t^*) - x(t_k)\|_2\big)\\
    & \leq (t-t_k)\big(- \sigma_{\min}(Q) + c \Delta \cdot \sigma_{\max}(P) (\sigma_{\max}(A) + \sigma_{\max}(A_1))^2\big) \\ 
    & \quad \quad \|x(t_k)\|_2 \max\limits_{s \in [t_k, t_{k+1}]} \|x(s)\|_2,      \label{eq:lyapunov_theory_block5}
\end{aligned}
\vspace{\squeeze}
\end{equation}
where $c > 0$ is an appropriate constant. In the last inequality, we apply Weyl's inequality to separate $\sigma_{\max}(A + A_1) \leq \sigma_{\max}(A) + \sigma_{\max}(A_1)$ and substitute the bound on $\|x(t^*) - x(t_k)\|_2$ in Eq.~\eqref{eq:lyapunov_theory_block4} to obtain the square term $(\sigma_{\max}(A) + \sigma_{\max}(A_1))^2$. In order for $V(x(t)) - V(x(t_k))$ to have a sufficient decrease, e.g. for some $d > 1$ ($d = 2$ in~\citet{clarke2010discontinuous}), 
\begin{equation}
    V(x(t)) - V(x(t_k)) \leq -(t-t_k) \frac{\sigma_{\min}(Q)}{d} \|x(t_k)\|_2 \max\limits_{s \in [t_k, t_{k+1}]} \|x(s)\|_2,
    \label{eq:lyapunov_theory_block6}
\vspace{\squeeze}
\end{equation}
the following gives a sufficient condition
\begin{equation*}
    \begin{aligned}
        & c \Delta \cdot \sigma_{\max}(P) (\sigma_{\max}(A) + \sigma_{\max}(A_1))^2 \leq \frac{d-1}{d}\sigma_{\min}(Q) \\
        \Leftrightarrow \;\; & \Delta \; \leq c' \frac{\sigma_{\min}(Q)}{\sigma_{\max}(P) (\sigma_{\max}(A) + \sigma_{\max}(A_1))^2}, 
        \label{eq:lyapunov_theory_block7}
    \end{aligned}
\end{equation*}
for some $c' > 0$.
\end{proof}
While the above bound can be loose due to the worst case singular-value bounds, it still provides a way to qualitatively analyze the system. As an interpretation, let us suppose $P = I$ results in $Q > 0$. Then loosely speaking, an unstable uncontrolled system $A$ with larger $\sigma_{max}(A)$ makes the bound smaller. The contribution of the control is more complicated with a trade-off involved: on one hand, the larger control makes the continuous controlled system $A + A_1$ more stable, increasing the $\sigma_{\min}(Q)$ term in the numerator; on the other hand, it also increases $\sigma_{\max}(A_1)$ and hence increases the denominator. 

\subsection{A Lyapunov-Krasovskii functional}
\label{sec:lyapunov3}
With the key observation that the piecewise constant control system aligns perfectly with the sample-data system framework, we seek to find a tighter lower bound on the hold limit using theory developed for sample-data systems. Several studies~\cite{fridman2014tutorial, fridman2001new, liu2012wirtinger} view the sample-data system as a special case of the time-delay system with delay $\tau(t) = t - t_k$, which has a constant rate of change $\dot{\tau}(t) = 1$ for all $t$. Lyapunov-Krasovskii functionals are commonly used to analyze the performance of time-delay systems, and naturally extend to the sample-data system (\ref{eq:matrix_system}), which can be equivalently written in the form $\dot{x}(t) = (A + A_1) x(t) - A_1 \int_{t_k}^{t} \dot{x}(s) ds$, as $x(t_k) = x(t) - \int_{t_k}^{t} \dot{x}(s) ds$. \citet{fridman2014tutorial} proposes the following Lyapunov-Krasovskii functional for sample-data system 
\vspace{\squeeze}
\[V(t, x(t), \dot{x}(t)) = x^\intercal(t) P x(t) + (\Delta - \tau(t)) \int_{t - \tau(t)}^{t}\dot{x}^\intercal(s) U \dot{x}(s) ds , \vspace{\squeeze}\]
where $\tau(t) = t - t_k$, $P > 0$, $U > 0$. The first term $x^\intercal(t) P x(t)$ in the above functional is the regular Lyapunov function for the unperturbed nominal system $\dot{x}(t) = (A + A_1) x(t)$, whereas the second integral term handles the integral perturbation $- \int_{t_k}^{t} \dot{x}(s) ds$. Jensen's inequality, descriptor method~\cite{fridman2001new}, and state-augmentation with $\eta_1(t) = col\{x(t), \dot{x}(t), \frac{1}{\tau(t)}\int_{t-\tau(t)}^{t}\dot{x}(s)ds\}$ are applied to arrive at the following proposition on a given hold length $\Delta$ with Linear Matrix Inequalities (LMIs):
\begin{proposition} \label{prop:lk} Let there exist $n \times n$ matrices $P > 0, U > 0$; $P_2$ and $P_3$ such that the LMIs (\ref{eq:Lyapunov_Krasovskii}) are feasible. Then (\ref{eq:matrix_system}) is asymp. stable for all variable sampling instants $t_{k+1} - t_{k} \leq \Delta$.
\vspace{\squeeze}
\begin{equation}
\begin{aligned}
    & \begin{bmatrix}
    \Phi_{11} & P - P_2^\intercal + (A + A_1)^\intercal P_3 \\
    * & -P_3 - P_3^\intercal + \Delta U \\
    \end{bmatrix} < 0, \\
    & \begin{bmatrix}
    \Phi_{11} & P - P_2^\intercal + (A + A_1)^\intercal P_3 & -\Delta P_2^\intercal A_1\\
    * & -P_3 - P_3^\intercal & -\Delta P_3^\intercal A_1 \\
    * & * & -\Delta U \\
    \end{bmatrix} < 0.
    \label{eq:Lyapunov_Krasovskii}
\end{aligned}
\vspace{\squeeze}
\end{equation}
where $\Phi_{11} = P_2^\intercal (A + A_1) + (A + A_1)^\intercal P_2$ and $*$ denotes the symmetric elements of the symmetric matrix.
\end{proposition}
\begin{proof} 
See~\cite{fridman2001new}.
\end{proof}

While the previous Lyapunov analysis provides upper bounds on the perturbations by bounding the actions of the linear operators using the maximum singular values (e.g. $\sigma_{max}(A)$, $\sigma_{max}(A_1)$), the Lyapunov-Krasovskii bound solves for matrices $P$, $U$, $P_2$, $P_3$ to account for the interactions among $A$,  $A_1$, and $A + A_1$, and hence can yield a tighter bound.

Additionally, while the above proposition takes a fixed controller $K$ as given to verify if such a controller can stabilize the system with a hold length $\Delta$, we can optimize for the controller $K$ using the following corollary that takes the sample-data system property into consideration.
\begin{corollary} \label{cor:lk}
Let there exist $n \times n$ matrices $\bar{P} > 0$, $\bar{U} > 0$, $Q$ and an $n_u \times n$-matrix $L$ and a tuning parameter $\epsilon$ such that the LMIs (\ref{eq:Lyapunov_Krasovskii_control}) are feasible, where $\bar{\Phi}_{11} = Q^\intercal A^\intercal + A Q + BL + L^\intercal B^\intercal$. Then (\ref{eq:matrix_system}) is asymp. stable for all variable sampling instants $t_{k+1} - t_{k} \leq \Delta$ with the stabilizing gain given by $K = L Q^{-1}$.
\begin{equation*}
    \begin{bmatrix}
    \bar{\Phi}_{11} & \bar{P} - Q + \epsilon Q^\intercal A^\intercal + L^\intercal B^\intercal\\
    * & -\epsilon(Q + Q^\intercal) + \Delta \bar{U} \\
    \end{bmatrix} < 0, 
\end{equation*}
\begin{equation}
    \begin{bmatrix}
    \bar{\Phi}_{11} & \bar{P} - Q + \epsilon (Q^\intercal A^\intercal + L^\intercal B^\intercal) & -\Delta BL\\
    * & -\epsilon(Q + Q^\intercal) & -\Delta \epsilon BL \\
    * & * & -\Delta \bar{U} \\
    \end{bmatrix} < 0,
    \label{eq:Lyapunov_Krasovskii_control}
\end{equation}
\end{corollary}
\begin{proof} 
From above and following~\cite{liu2012wirtinger}, we can perform full state-feedback controller design by substituting $P_3 = \epsilon P_2$ where $\epsilon$ is a tuning parameter, $Q = P_2^{-1}$, $\bar{P} = Q^\intercal P Q$, $\bar{U} = Q^\intercal U Q$ and $L = K Q$. Multiplying LMIs (\ref{eq:Lyapunov_Krasovskii_control}) by $diag\{Q^\intercal , ..., Q^\intercal\}$ and $diag\{Q, ..., Q\}$ from the left and right, we recover LMIs (\ref{eq:Lyapunov_Krasovskii}). \end{proof}
\vspace{\squeeze}
\vspace{\squeezeSection}
\section{Extensions to Human Errors}
In practice, human drivers may be unable to comply with guidance perfectly, leading to magnitude error and reaction delay. This section presents extensions of the theory to incorporate these additional human driving behaviors.
\vspace{\squeezeSection}
\subsection{Magnitude error}
When presented with an instruction, human drivers are prone to magnitude errors and may only be capable of following the instruction within a perturbed range. We modify the error dynamics of the linearized piecewise-constant control system in Eq.~(\ref{eq:matrix_system}) as follows
\vspace{\squeeze}
\begin{equation}
    \dot{x}(t) = A x(t) + A_1 x(t_k) + B_d d(t), \quad k = 0, 1, ...
    \label{eq:magnitude_error_system}
\vspace{\squeeze}
\end{equation}
where $B_d = diag(\{0, 1, ..., 0, 1\})$ due to perturbations in acceleration for both the guided and unguided HVs, and $d = [d_1(t), ..., d_n(t)]^\intercal \in \mathbb{R}^n$ is the perturbation vector. 
\begin{assumption}
We consider two scenarios 1) \textit{nonvanishing perturbation}: $\|d(t)\|_2 \leq \bar{d}_{nv}$ for some constant $\bar{d}_{nv} \geq 0$, which assumes the human driver incurs a nonvanishing error whose magnitude upper bound is independent of the system's error state; 2) \textit{vanishing perturbation}: $\|d(t)\|_2 \leq \bar{d}_{v}\|x(t)\|_2$ for some constant $\bar{d}_{v} \geq 0$, which assumes the upper bound of the human driver's magnitude error is proportional to the system's error state and diminishes as the error state approaches equilibrium.
\label{assump:nv_and_v}
\end{assumption}
\vspace{-0.5cm}
\subsubsection{Lyapunov Analysis}
\begin{proposition}\label{prop:magnitude_error_lyapunov}
Under the assumptions of  Prop.~\ref{prop:lyapunov} and nonvanishing perturbation, the system under magnitude error in Eq.~\eqref{eq:magnitude_error_system} eventually converges within a bounded region around the equilibrium (the ultimate bound) where
\vspace{\squeeze}
\begin{equation}
    \|x(t)\|_2 \leq \frac{10 \sigma_{max}(B_dP + P B^\intercal_d)}{1/d\sigma_{\min}(Q)}\bar{d}_{nv}
\vspace{\squeeze}
\end{equation}
for some $d > 1$, under the same hold length $\Delta \leq c' \frac{\sigma_{\min}(Q)}{\sigma_{\max}(P) (\sigma_{\max}(A) + \sigma_{\max}(A_1))^2}$ as the system without magnitude error in Prop.~\ref{prop:lyapunov}.

Assume vanishing perturbation, the system under magnitude error in Eq.~\eqref{eq:magnitude_error_system} is asymptotically stable for hold length
\vspace{\squeeze}
\begin{equation}
\begin{aligned}
    \Delta \; \leq c' \frac{\sigma_{\min}(Q) - 5 \sigma_{max}(B_dP + P B^\intercal_d)\bar{d}_v}{\sigma_{\max}(P) (\sigma_{\max}(A) + \sigma_{\max}(A_1))^2}
\end{aligned}
\vspace{\squeeze}
\end{equation}
\end{proposition}
\begin{proof} The proof follows that of Prop.~\ref{prop:lyapunov}, with additional perturbation terms in Eq.~\eqref{eq:lyapunov_theory_block1} to account for the magnitude error. These perturbation terms are similarly bounded by maximum singular value bounds as in the proof of Prop.~\ref{prop:lyapunov}. See Appendix A1 in the extended article for details.
\end{proof}

\vspace{\squeezeSection}
\subsubsection{Lyapunov-Krasovskii Analysis}

We can extend the Lyapunov Krasovskii analysis with $H_\infty$ robust control to accommodate magnitude error in the system, following a similar derivation in a previous work~\cite{li2014stabilization}. 
\begin{definition} A system is robust at $H_\infty$ disturbance attenuation level $\gamma > 0$ if $\left(\int_{0}^{\infty} x(t)^\intercal x(t)\right)^{1/2} \leq  \gamma \left(\int_{0}^{\infty} d(t)^\intercal d(t)\right)^{1/2}$.
\label{def:Hinf_robust}
\end{definition}
\begin{proposition} \label{prop:magnitude_error_lk} Under the assumptions in Prop.~\ref{prop:lk}, the system under magnitude error~(\ref{eq:magnitude_error_system}) is asymptotically stable at a $H_\infty$ disturbance attenuation level $\gamma > 0$ for all varying sampling instants $t_{k+1} - t_k \leq \Delta$ if the following LMIs are feasible.
\begin{equation}
\begin{aligned}
    & \begin{bmatrix}
    \Phi_{11} + I & P - P_2^\intercal + (A + A_1)^\intercal P_3 & P_2^\intercal B_d \\
    * & -P_3 - P_3^\intercal + \Delta U & P_3^\intercal B_d \\
    * & * & -\gamma^2 I
    \end{bmatrix} < 0, \\
    & \begin{bmatrix}
    \Phi_{11}  + I & P - P_2^\intercal + (A + A_1)^\intercal P_3 & -\Delta P_2^\intercal A_1 & P_2^\intercal B_d\\
    * & -P_3 - P_3^\intercal & -\Delta P_3^\intercal A_1 & P_3^\intercal B_d\\
    * & * & -\Delta U  & 0\\
    * & * & 0 & -\gamma^2 I
    \end{bmatrix} < 0.
    \label{eq:magnitude_error_Lyapunov_Krasovskii}
\end{aligned}
\end{equation}
\end{proposition}
\begin{proof} In the presence of the magnitude error~(\ref{eq:magnitude_error_system}), we can first employ a similar derivation as in~\citet{fridman2014tutorial} to derive LMIs akin to Eq.~\eqref{eq:Lyapunov_Krasovskii}, with an expanded augmented state $\eta_2 = col\{x(t), \dot{x}(t), v_1, d(t)\}$ that incorporates the magnitude error. Then, we follow~\citet{li2014stabilization} to modify the LMIs for $H_{\infty}$ robustness and conclude with the proposition. See Appendix A3 in the extended article for details. 
\vspace{\squeezeSection}
\end{proof}
\vspace{0.1cm}
\subsection{Reaction delay}
Human drivers are subject to reaction delay upon receiving instructions. For each $k = 0, 1, ...$, we consider a delayed time period $[t_k + \sigma(t_k), \;t_{k+1} + \sigma(t_{k+1})] :\triangleq [t_{k'}, t_{k'+1}]$ where $0 \leq \sigma(t) \leq \Sigma$ is the reaction delay from the human driver. The dynamics within the delayed time period is 
\vspace{\squeeze}
\begin{equation}
    \dot{x}(t) = A x(t) + A_1 x(t_k),  \label{eq:reaction_delay_system}
\vspace{\squeeze}
\end{equation}
where the full state feedback control $A_1 x(t_k) = - BK x(t_k)$ is determined based on the error state $x(t_k)$ at a time $t_k$ prior to the start time $t_{k'}$ of the delayed time period.
\vspace{0.07cm}
\subsubsection{Lyapunov Analysis}

\begin{proposition} \label{prop:reaction_delay_lyapunov} Under the assumptions in Prop.~\ref{prop:lyapunov} and the assumption where for some constant $\bar{D}_v \geq 0$, $\|\dot{x} (s)\|_2 \leq \bar{D}_{v} \max\limits_{s \in [t_{k'}, t_{k'+1}]}\|\dot{x}(t_{s})\|_2 \quad \forall s \in [t_{k}, t_{k'}]$, the system under reaction delay~(\ref{eq:reaction_delay_system}) is asymptotically stable for hold length 
\begin{equation}
    \Delta \; \leq c' \frac{\sigma_{\min}(Q) }{\sigma_{\max}(P) (\sigma_{\max}(A) + \sigma_{\max}(A_1))^2} - c'' \bar{D}_v \Sigma
\end{equation}
up to scaling constants $c' > 0, c'' > 0$.
\end{proposition}
\begin{proof} See Appendix A2 in the extended article for details. The proof follows that of Prop.~\ref{prop:lyapunov}, with an additional perturbation term in Eq.~\eqref{eq:lyapunov_theory_block1} to account for reaction delay. 
\end{proof}
Notably, as the proof treats the reaction delay period $[t_k, t_k + \sigma(t_k)]$ as adding noise to the system and bounds it by a maximum singular value bound, the analysis naturally applies to a broader class of human driving behaviors: for example, instead of instantly switching from the control $A_1 x(t_{k-1})$ to $A_1 x(t_k)$, the driver may transition smoothly in between, and the transition period can be considered as a delay period.

\subsubsection{Lyapunov-Krasovskii Analysis}
Similarly, we can adapt the original proof~\cite{fridman2014tutorial} of Prop.~\ref{prop:lk} to incorporate reaction delay by modifying the Lyapunov-Krasovskii functional as: 
\vspace{\squeeze}
\begin{equation}
    \resizebox{.49\textwidth}{!}{
    $V(t, x(t), \dot{x}(t)) = x^\intercal(t) P x(t) + (\Delta + \Sigma - \tau(t)) \int_{t - \tau(t)}^{t}\dot{x}^\intercal(s) U \dot{x}(s) ds, $
    }
\end{equation}
where $\Sigma$ is an upper bound on the reaction delay. Instead of a direction extension, we may further obtain a tighter bound on the hold limit by combining the above sample-data Lyapunov-Krasovskii functional with a time-delay Lyapunov-Krasovskii functional such as the one presented in~\citet{li2014stabilization}. We leave this as a future work.

\color{black}
\section{Numerical analysis}
\label{sec:experiment}
In this section, we compare the Lyapunov analysis and the Lyapunov-Krasovskii analysis with the hold limit from empirical simulation. We aim to answer the following questions:
\begin{enumerate}
    \item How well does the theory match simulation? To what extent do \textit{simplified} theoretical analyses explain integrated traffic flow stability under coarse-grained guidance?
    \item What relationships emerge from the problem parameters and how do they affect stability?
    \item Can we derive better piecewise-constant controllers using the Lyapunov or Lyapunov-Krasovskii analysis?
\end{enumerate}

\vspace{-0.2cm}
\subsection{Experimental Setup and Results on the default parameters}
\label{sec:experiment_setup}
We adopt the implementation from~\citet{zheng2020smoothing} in Python and extend it to the piecewise-constant control setting. A summary of all parameters and their default values is listed in Table~\ref{tab:params}. Vehicles are initialized by a uniform perturbation around the equilibrium, with the $i^{th}$ vehicle's position and velocity $(x^i_0, v^i_0) = (is^* + \delta_s, v^* + \delta_v)$ where $\delta_s \sim Unif[-7.5, 7.5], \delta_v \sim Unif[-4.5, 4.5]$, and $v^* = V(s^*)$ from Eq.~\eqref{eq:vclipping} is the equilibrium velocity corresponding to the equilibrium spacing $s^* = L / n$. By default, we apply the same $\mathcal{H}_2$ optimal full state-feedback controller for the continuous system to the sample-data system by holding it piecewise-constant. The controller 
\vspace{\squeeze}
\begin{equation}
    u(t) = -K x(t),
\vspace{\squeeze}
\end{equation}
where $K \in \mathbb{R}^{1\times 2n}$, can be obtained by the following convex program with $K = ZX^{-1}$:
\begin{equation}
\begin{aligned}
   \min\limits_{X, Y, Z} \quad &  \text{Trace}(QX) + \text{Trace}(RY) \label{eq:H2_obj}\\
   \text{subject to} \quad & (AX - BZ) + (AX - BZ)^\intercal +H H ^\intercal \preccurlyeq 0,\\
   & \begin{bmatrix}
        Y & Z\\
        Z^\intercal & X
     \end{bmatrix} \succcurlyeq 0, X \succ 0.
\end{aligned}
\end{equation} 
\begin{equation}
    \text{where } \quad Q^{\frac{1}{2}} = \text{diag}(\gamma_s, \gamma_v, ..., \gamma_s, \gamma_v), \; R^{\frac{1}{2}} = \gamma_u, \; H = I \label{eq:H2_param}
\end{equation}
with the default $\gamma_s = 0.03, \gamma_v = 0.15, \gamma_u = 1$, corresponding to the performance state $ z(t) = \begin{bmatrix}Q^{\frac{1}{2}}\\ 0\end{bmatrix}x(t) + \begin{bmatrix}0 \\ R^{\frac{1}{2}}\end{bmatrix}u(t)$.

\begin{table*}[!ht]
\centering
\begin{tabular}{ |c|c|c| } 
 \hline
 \textbf{Symbol} & \textbf{Default} & \textbf{Description} \\ 
 \hline
 \multicolumn{3}{|c|}{\textbf{\textit{System parameters}}} \\
 \hline
 $L$ & $400m$ & Circumference of the ring-road, where the equilibrium spacing $s^* = L / n$ \\ 
 \hline
 $n$ & $20$ & Number of vehicles in the ring-road system, where the equilibrium spacing $s^* = L / n$ \\
 \hline
 $s_{st}$ & $5m$ & Small spacing threshold such that the optimal velocity $= 0$ below the threshold, see Eq.~\eqref{eq:vclipping} \\
 \hline
 $s_{go}$ & $35m$ & Large spacing threshold such that the optimal velocity $= v_{max}$ above the threshold, see Eq.~(\ref{eq:vclipping}) \\
 \hline
  $v_{max}$ & $30m/s$ & Maximum optimal velocity, see Eq.~\eqref{eq:vclipping} and (\ref{eq:ovm_f}) \\
 \hline
 $\alpha$ & $0.6$ & Driver's sensitivity to the difference between the current velocity and the desired spacing-dependent optimal velocity, see Eq.~\eqref{eq:ovm}\\
 \hline
 $\beta$ & $0.9$ & Driver's sensitivity to the difference between the velocities of the ego vehicle and the preceding vehicle, see Eq.~\eqref{eq:ovm} \\
 \hline
 \multicolumn{3}{|c|}{\textbf{\textit{Control parameters}}}\\
 \hline
 $k_{mult}$ & $1$ & Scale the $\mathcal{H}_2$ optimal controller $K_{cont}$ by a constant: $K_{new} = k_{mult} \cdot K_{cont}$\\
 \hline
 $\gamma_s$ & $0.03$ & weight on the position derivation from equilibrium in the $\mathcal{H}_2$ optimal control objective, see Eq.~\eqref{eq:H2_obj} and Eq.~\eqref{eq:H2_param}\\
 \hline
 $\gamma_v$ & $0.15$ & weight on the velocity derivation from equilibrium in the $\mathcal{H}_2$ optimal control objective, see Eq.~\eqref{eq:H2_obj} and Eq.~\eqref{eq:H2_param} \\
 \hline
 $\gamma_u$ & $1$ & weight on the control magnitude in the $\mathcal{H}_2$ optimal control objective, see Eq.~\eqref{eq:H2_obj} and Eq.~\eqref{eq:H2_param} \\
 \hline
\end{tabular}
 \caption{System and control parameters for the traffic system based on the Optimal Velocity Model.}
 \label{tab:params}
 \vspace{-0.3cm}
\end{table*}

\textbf{Simulation stability criteria:} \label{sec:setup_stability_criteria}
We simulate the system by integrating the ordinary differential equation (\ref{eq:matrix_system}) using the forward Euler method, with a discretization of $T_{step}=0.01s$. We say a system (uncontrolled, or with continuous / piecewise-constant control) is \textit{stable in simulation} if (1) $50$ simulated trajectories from different initial perturbations all converge to the equilibrium within $TotalTime = 300s$, and (2) no vehicle collides within the trajectory (given by negative spacings). To mitigate collisions, we follow~\citeauthor{zheng2020smoothing} to equip all vehicles with a standard automatic emergency braking system $\dot{v}(t) = a_{min}, \text{ if } \frac{v_i^2(t) - v_{i-1}^2(t)}{2(s_i(t) - s_d)} \geq |a_{min}|$, where $a_{min} = -5m/s^2$ is the maximum deceleration rate of each vehicle, and $s_d = 0.5m$ is the safe distance. 

The simulations empirically identify stable hold times without collisions. We denote the longest of these hold times for each parameter setting as the simulation hold limit, which is an approximation to the true hold limit without collision. Notably, as the theoretical analysis in this work focuses on certifying asymptotic stability (i.e. convergence of the trajectories to equilibrium), the analysis do not model the collision constraints and hence may overestimate the simulation hold limit when the system converges yet collisions occur. A future work involves combining Lyapunov-Krasovskii LMIs with control barrier functions~\cite{ames2016control} to explicitly model collisions and ensure both asymptotic stability and safety.

Moreover, rather than focusing on the actual values of the hold limits, we utilize the hold limits for comparative analysis, where we compare parameter settings and identify regimes in which system and control parameters enable longer hold limits.

\begin{figure}[ht!]
    \centering
    \subfigure[]{\includegraphics[width=0.23\textwidth]{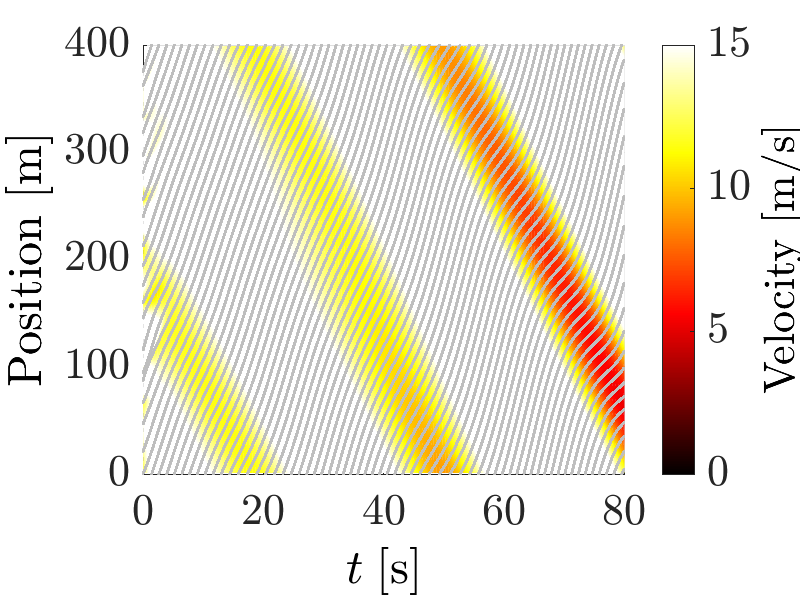}} 
    \subfigure[]{\includegraphics[width=0.23\textwidth]{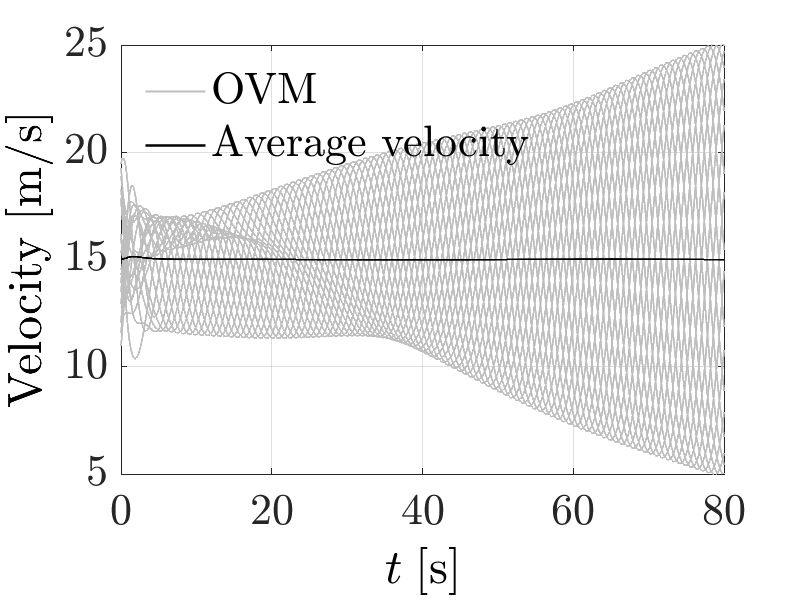}} 
    \caption{The traffic system with all unguided HVs under OVM and no controlled vehicle is unstable under the default parameters in Sec.~\ref{sec:experiment_setup}. The equilibrium spacing and velocity are $20m$ and $15m/s$. (a) The time-space diagram. Darker colors represent lower velocities. (b) The time-velocity diagram. The initial perturbation on the velocities get amplified, leading to the formation of stop-and-go waves in the system.}
    \label{fig:ovm_uncontrolled}
\end{figure}

\textbf{Simulation results with the default OVM parameters:} We study the behavior of the system by putting a piecewise-constant hold on the controller for $\Delta \gg T_{step}$ seconds. Without any controlled vehicle, the default OVM system is unstable (see Fig.~\ref{fig:ovm_uncontrolled}), forming stop-and-go waves gradually. Zheng et al.~\cite{zheng2020smoothing} show that introducing one autonomous vehicle with the continuous $\mathcal{H}_2$ optimal controller is able to stabilize the \textit{continuous system}. In Fig.~\ref{fig:H2_controlled}, we show the behavior of the sample-data traffic system by holding the same controller for $\Delta = 1.59s$ (left) and $\Delta = 2.29s$ (right). With a smaller hold length of $1.59s$, the controller is able to stabilize the system. However, with a slightly larger hold length of $2.29s$, we observe unstable system behavior, where holding the control piecewise-constant introduces an excessive amount of noise that breaks the system's stability. It is interesting to observe the sawtooth pattern in the time-velocity diagram in Fig.~\ref{fig:H2_controlled}d, where errors are accumulated within each holding period, but get corrected at the next holding period when we update the control. While there is system slowdown, the velocity perturbation is constrained within a range between $[7.5, 20]$ $m/s$, instead of getting amplified and diverging as in Fig.~\ref{fig:ovm_uncontrolled}. 

\begin{figure*}[!]
\centering
    \subfigure[]{\includegraphics[width=0.23\textwidth]{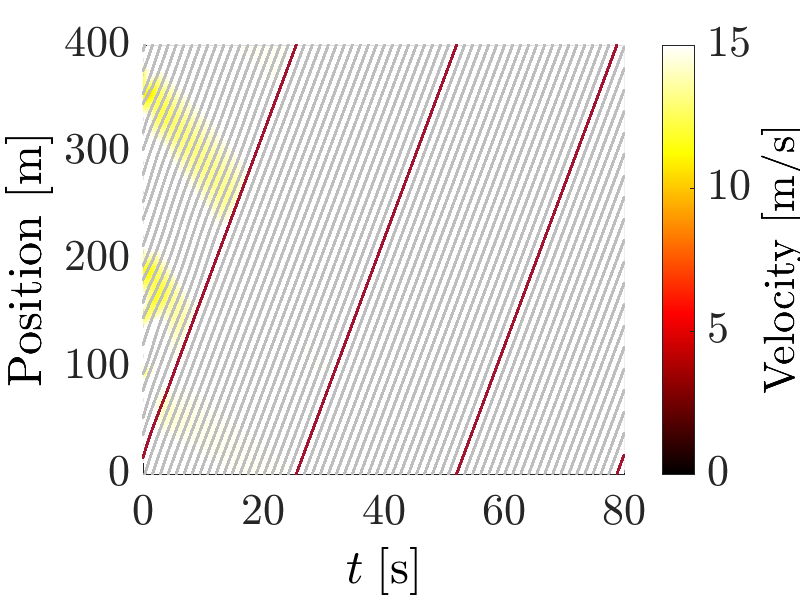}}
    \subfigure[]{\includegraphics[width=0.23\textwidth]{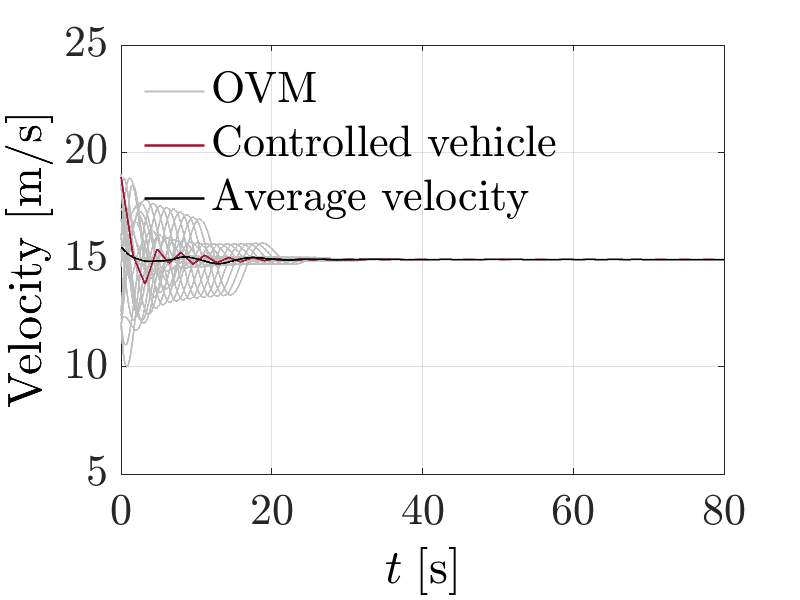}}
    \subfigure[]{\includegraphics[width=0.23\textwidth]{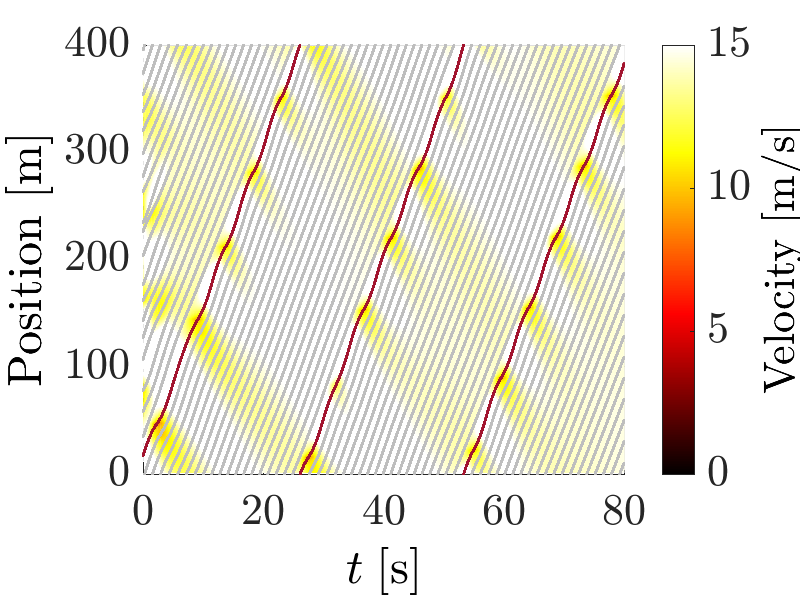}} 
    \subfigure[]{\includegraphics[width=0.23\textwidth]{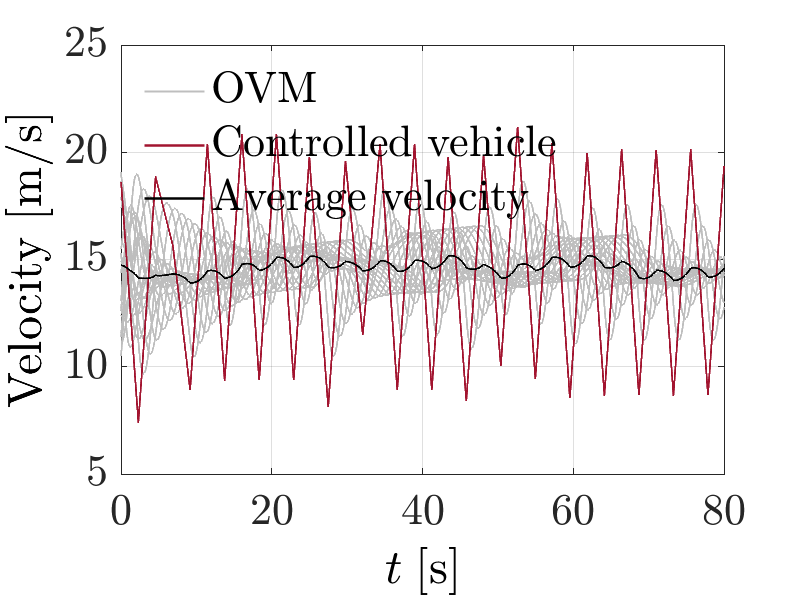}}
    \caption{The traffic system consists of $n-1$ unguided HVs under OVM (gray) and $1$ guided HV under piecewise-constant control (red) with different hold lengths $\Delta$s. The guided HV applies the same $\mathcal{H}_2$ optimal control gain matrix for the continuous system to the sample-data system, under the default parameters in Sec.~\ref{sec:experiment_setup}. (a) and (b): The time-space and time-velocity diagrams when $\Delta = 1.59s$. The traffic is stabilized to the equilibrium velocity $15m/s$ after a short amount of time. (c) and (d): The time-space and time-velocity diagrams when $\Delta = 2.29s$. The system becomes unstable when the hold length is too long. \vspace{-0.2cm}}
    \label{fig:H2_controlled}
\end{figure*}

\vspace{\squeezeSection}
\vspace{0.3cm}
\subsection{How well does the theory match simulation?}
\label{sec:experiment2}
Fig.~\ref{fig:sensitivity} provides the results of varying seven OVM system parameters $(L, n, v_{max}, s_{st}, s_{go}, v_{max}, \alpha, \beta)$ and three control parameters $(k_{mult}, \gamma_s, \gamma_v)$ (see Table~\ref{tab:params} for notation), comparing the theoretical hold limit estimates of Eq.~\eqref{eq:lyapunov_theory} and~\eqref{eq:Lyapunov_Krasovskii} with simulation hold limits. Ten analyses serve as a sensitivity analysis where we vary one parameter while fixing the others to default values. In the last analysis, we vary $(\alpha, \beta, s_{st}, s_{go})$ simultaneously while fixing the rest to default: we increase $s_{st}$ and decrease $s_{go}$, making the optimal velocity curve steeper (See Fig.~\ref{fig:interp}). Meanwhile, we decrease $\alpha$ and increase $\beta$ to allow human drivers to focus on the preceding vehicle when the optimal velocity becomes challenging to follow. The hold limit information identifies parameter regimes to guide the design of traffic systems and controllers for more effective coarse-grained guidance. For each analysis, we perform a binary search within $[0s, 10s]$ with a granularity of $T_{step} = 0.01s$ in simulation to find the empirical hold limit. We solve for a continuous $\mathcal{H}_2$ optimal controller using the respective system and control parameters. The scale of the $y\text{-}axis$ for Lyapunov analysis and OVM stability are given in Table~\ref{tab:scale}.

\begin{figure*}
    \centering
    \includegraphics[width=1.0\textwidth]{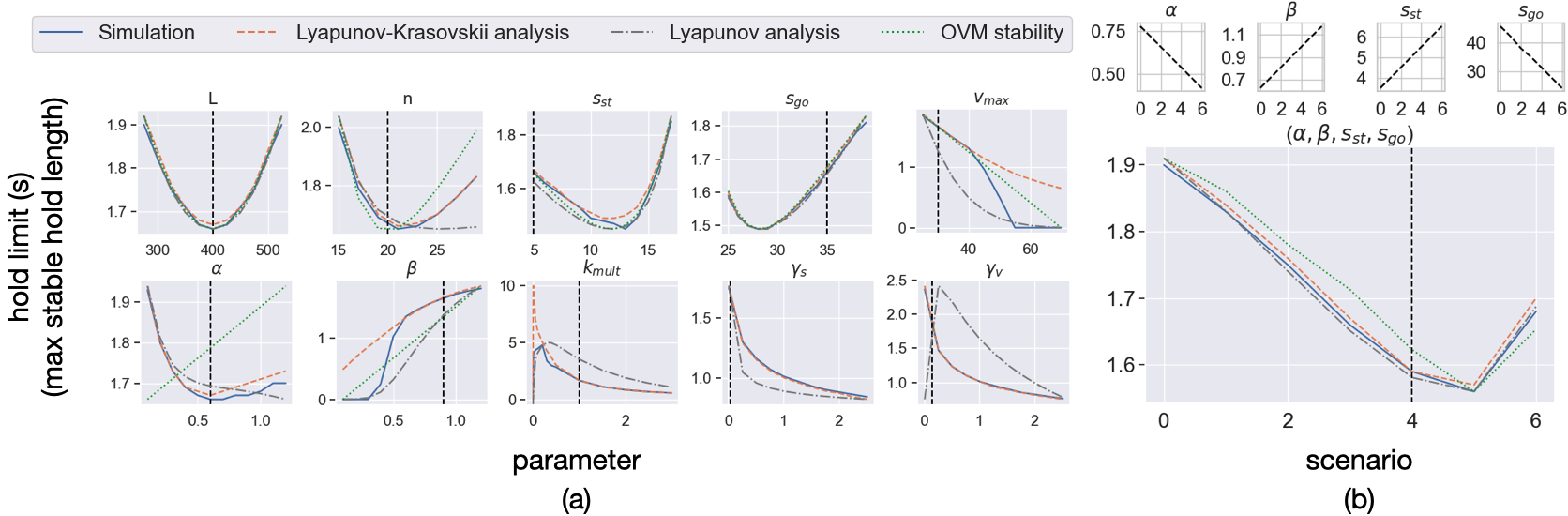}  
\caption{The hold limit that stabilizes the system from simulation (solid blue), Lyapunov-Krasovskii analysis in Eq.~\eqref{eq:Lyapunov_Krasovskii} (dashed orange), Lyapunov analysis in Eq.~\eqref{eq:lyapunov_theory} (dash-dotted gray), and uncontrolled OVM stability criterion in Eq.~\eqref{eq:string_ovm} (dotted green). Default parameter values are shown as the black vertical lines in each plot. The $y$-axis scale is for simulation and Lyapunov-Krasovskii analysis. (a) We vary each of the ten parameters one at a time while keeping the rest at default values. The $x$-axis represents the parameter value, and the $y$-axis scale for the Lyapunov analysis and the OVM stability are displayed in Table~\ref{tab:scale}. (b) We vary $(\alpha, \beta, s_{st}, s_{go})$ simultaneously, where the top four mini plots show the parameter values corresponding to each of the seven joint value sets. The $x$-axis represents the index of each value set (each scenario). The $y$-axis scale for the Lyapunov analysis is $(8.35\times10^{\text{-}4}, 1.80\times10^{\text{-}3})$, and for OVM stablity is $(\text{-}1.32, \text{-}1.91\times10^{\text{-}2})$. \vspace{-0.2cm}}
    \label{fig:sensitivity}
\end{figure*}

\begin{table}[!ht]
\centering
\begin{tabular}{ |c|c|c| } 
 \hline
 \textbf{Symbol} & \textbf{Lyapunov analysis} & \textbf{OVM stability} \\
 \hline
 \multicolumn{3}{|c|}{\textbf{\textit{System parameters}}} \\
 \hline
 $L$ & $(1.07\times10^{\text{-}3}, 1.81\times10^{\text{-}2})$ & $(\text{-}7.74\times10^{\text{-}1}, \text{-}5.99\times10^{\text{-}2})$ \\ 
 \hline
 $n$ & $(7.70\times10^{\text{-}4}, 3.19\times10^{\text{-}3})$ & $(\text{-}7.62\times10^{\text{-}1}, 2.94\times10^{\text{-}2})$\\
 \hline
 $s_{st}$ & $(6.82\times10^{\text{-}4}, 1.68\times10^{\text{-}3})$ & $(\text{-}1.29, \text{-}1.67\times10^{\text{-}1})$\\
 \hline
 $s_{go}$ & $(6.88\times10^{\text{-}4}, .56\times10^{\text{-}3})$ & $(\text{-}1.28, \text{-}2.79\times10^{\text{-}1})$ \\
 \hline
  $v_{max}$ & $(6.79\times10^{\text{-}5}, 1.72\times10^{\text{-}3})$ & $(\text{-}5.17, 1.76\times10^{\text{-}2})$\\
 \hline
 $\alpha$ & $(9.65\times10^{\text{-}4}, 1.89\times10^{\text{-}3})$ & $(\text{-}1.30, \text{-}8.66\times10^{\text{-}2})$\\
 \hline
 $\beta$ & $(7.44\times10^{\text{-}5}, 1.56\times10^{\text{-}3})$ & $(\text{-}2.45, \text{-}3.16\times10^{\text{-}2})$\\
 \hline
 \multicolumn{3}{|c|}{\textbf{\textit{Control parameters}}}\\
 \hline
 $k_{mult}$ & $(0, 3.02\times10^{\text{-}3})$ & - \\
 \hline
 $\gamma_s$ & $(1.08\times10^{\text{-}4}, 1.47\times10^{\text{-}3})$ & - \\
 \hline
 $\gamma_v$ & $(5.60\times10^{\text{-}4}, 1.17\times10^{\text{-}3})$ & - \\
 \hline
\end{tabular}
 \caption{Scales of the theoretical hold limit estimates (the minimum and maximum of $y\text{-}axis$ in Fig.~\ref{fig:sensitivity} for Lyapunov analysis and OVM stability.)}
 \label{tab:scale}
\end{table}

Examining the effectiveness of piecewise-constant control under different traffic conditions is a complex problem. Following the motivation of ``All models are wrong, but some are useful,'' it is attractive to consider these reduced-order linearized models as proxies for analyzing the true traffic problem. We thus consider three theoretical approaches for estimating the hold limit:
\begin{enumerate}
    \item The Lyapunov analysis: see Eq.~\eqref{eq:lyapunov_theory}; we set $c' = 1$. Due to redundancy in headway representation with $\tilde{s}_1 + \tilde{s}_2 ... + \tilde{s}_n = 0$, we first obtain the reduced representation by omitting $\tilde{s}_1$ from the state vector and replacing it with $-\tilde{s}_2 - ... -\tilde{s}_{n}$ to construct the reduced system matrices $A^\dagger, B^\dagger, K^\dagger$. Then, we set $Q = I_{(n-1)\times (n-1)}$ which has $\sigma_{min}(Q) = 1$ and solve for $P$ from the Lyapunov equation $(A^\dagger - B^\dagger K^\dagger)P + P(A^\dagger - B^\dagger K^\dagger)^\intercal = -Q$ to obtain $\sigma_{max}(P)$ in the denominator of Eq.~\eqref{eq:lyapunov_theory}. 
    \item The Lyapunov-Krasovskii analysis: see LMIs~(\ref{eq:Lyapunov_Krasovskii}). We perform a binary search within $[0s, 10s]$ with a granularity of $T_{step} = 0.01s$ to find the theoretical hold limit estimate such that the LMIs are feasible.
    \item The OVM stability: stability theory of the linearized, \textit{uncontrolled} system. Previous work~\cite{cui2017stabilizing} uses string stability to analyze the linearized, uncontrolled continuous OVM model, and derive the stability criteria
    $\alpha + 2\beta \geq 2\dot{V}(s^*) = 2\dot{V}(L / n)$. Equivalently, for $s^* = L / n \in [s_{st}, s_{go}]$, the OVM system is stable if 
    \vspace*{-0.2cm}
    \begin{equation}
    \alpha + 2\beta - v_{max} \frac{\pi}{s_{go} - s_{st}} \sin\big(\pi \frac{L / n - s_{st}}{s_{go} - s_{st}}\big) \geq 0  
    \label{eq:string_ovm}
    \end{equation}
    We plot the value of the left hand side in Fig.~\ref{fig:sensitivity}, which takes on negative values because we choose parameter values so that the uncontrolled system is unstable.  We use this estimate as a continuous proxy of the instability level in the uncontrolled system. A higher level of instability in the uncontrolled system (indicated by a more negative left-hand side) likely results in a shorter controller hold limit, as the system may require more frequently updated controls for stabilization.
\end{enumerate}

\textbf{Overall findings:} We observe that (1) both OVM stability and the Lyapunov do generally capture the trends quite well, (2) the Lyapunov-Krasovskii analysis captures not only the trend but also the absolute hold limit, indicating that the effect of linearizing the system is not a strong limitation of the approach, and (3) it is important to consider both the role of the controller (insufficiency of OVM stability to capture the trend, particularly in the case of $\alpha$) and the effect of the Lyapunov-Krasovskii integral (inadequacy of Lyapunov analysis to obtain the correct absolute scale). 

\vspace{-0.05cm}
\textbf{Lyapunov-Krasovskii Analysis:} the Lyapunov-Krasovskii analysis shares the same scale as the simulation, which is depicted as the numbers on the left of the $y\text{-}axis$. The Lyapunov-Krasovskii analysis is remarkably accurate in general, matching both the trend of the simulation and \textit{the absolute scale} of all parameters, whereas the other two theoretical methods only provide relative trend estimates. The Lyapunov-Krasovskii analysis overestimates the simulation hold limits for large $v_{max}$, small $\beta$, and small $k_{mult}$, however, where the unstable uncontrolled system results in collisions not modeled by the LMIs~(\ref{eq:Lyapunov_Krasovskii}), as discussed in Sec.~\ref{sec:setup_stability_criteria}. In such cases, the Lyapunov analysis gives a more accurate bound by more aggressively penalizing the worst-case behavior given by $\sigma_{max}(A)$ (unstable uncontrolled system) or $\sigma_{min}(Q)/\sigma_{max}(P)$ (controller with small magnitude).

\textbf{Lyapunov Analysis:} The Lyapunov analysis matches the trend of the simulation hold limits decently well, despite with smaller absolute scale than the simulation. While the worst-case singular value bounds in the Lyapunov analysis allow a more conservative estimate than Lyapunov-Krasovskii for large $v_{max}$ and small $\beta$, they become overly aggressive for large $n$ and $\alpha$. In such cases, Lyapunov-Krasovskii provides a better estimate by considering the interaction of $A$ (the uncontrolled system), $BK$ (the control) and $A-BK$ (the controlled system). Regarding $(k_{mult}, \gamma_s, \gamma_v)$, the Lyapunov analysis generally captures the correct trend, but with discrepancies in the absolute slopes or peaks. Since the analysis holds up to a scaling constant, the slope and peak location can vary depending on different scalings of $\sigma_{max}(A)$ and $\sigma_{max}(A_1)$. We keep equal scaling in the analysis for clarity of interpretation, and leave finding more accurate scalings to future work.

\vspace{-0.05cm}
\textbf{Uncontrolled OVM:} To our surprise, the uncontrolled OVM stability matches the trend of the simulation hold limits particularly well for a few parameters $(L, s_{st}, s_{go})$, and has only minor mismatch for $(v_{max}, \beta)$. As observed in the $L$ subplot in Fig.~\ref{fig:interp_subp}, in these cases, the trends of the uncontrolled system $\sigma_{max}(A)$ and the controller $\sigma_{max}(-BK)$ align well. However, discrepancies arise in the case of $n$ and particularly $\alpha$ (where opposite trends are observed), leading to misalignment between the simulation hold limits (with control) and the OVM stability (which solely considers the uncontrolled system $A$). In the misaligned cases, the impact of the controller that is captured by Lyapunov and Lyapunov-Krasovskii analysis is necessary for a more accurate trend estimate.

\begin{figure}[!tb]
    \centering
    \includegraphics[width=0.42\textwidth]{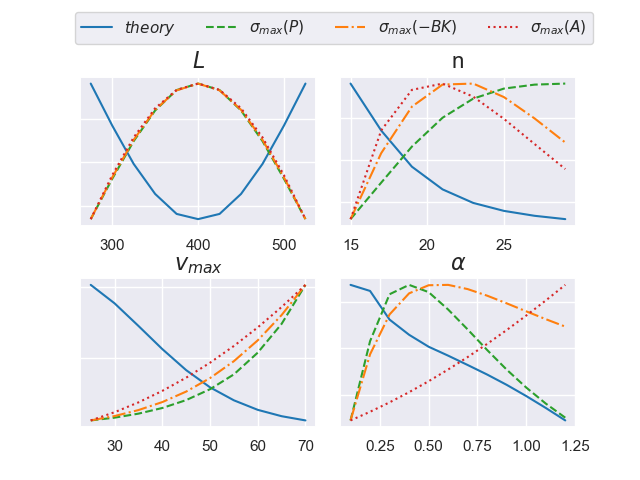}
    \caption{A visualization of components in the Lyapunov analysis (Eq. (\ref{eq:lyapunov_theory})) for four system parameters $L, n, v_{max}, \alpha$. We plot the denominator components $\sigma_{max}(P)$ that represents the continuous controlled system (dashed green), $\sigma_{max}(A)$ (dotted red) that represents the continuous uncontrolled system, $\sigma_{max}(A_1) = \sigma_{max}(-BK)$ that represents the control (dash-dotted orange), and the final theory bound on the hold limit $\Delta$ that stabilizes the system (solid blue). Note that numerator component $\sigma_{min}(Q) = 1$ by construction. The absolute scales of the different components are omitted. \vspace{-0.1cm}}
    \label{fig:interp_subp}
\end{figure}

\vspace{\squeezeSection}
\vspace*{0.2cm}
\subsection{How do traffic conditions affect the hold limit?}
\label{sec:interpretation}
In this section, we interpret relationships between system parameters $(L, n, s_{st}, s_{go}, v_{max}, \alpha, \beta)$, which represent different traffic conditions, and their respective hold limits. Overall, we observe three main types of traffic situations that promote longer hold limits by means of \textit{low driver sensitivity}: (1) traffic conditions (density, speed limit, and spacing thresholds) that promote a smoother spacing response, i.e., the flatter region of the optimal velocity function (through various combinations of $L, n, s_{go}, s_{stop}, v_{\max}$), (2) low sensitivity of drivers to relative position (low $\alpha$), and (3) high sensitivity of drivers to relative speed, which tends towards equilibrium (high $\beta$). 

\textbf{Smoother spacing response}: We observe that $(L, n, s_{st}, s_{go}, v_{max})$ determines various aspects of the optimal velocity function, as shown in Eq.~\eqref{eq:vclipping} and Fig.~\ref{fig:interp}. For example, the parameters $L$ and $n$ are related to the traffic density. Their ratio $s^* = L / n$ determines the equilibrium spacing, which in turn determines the desired optimal velocity $v^* = V(s^*)$, clipped within the range $[0, v_{max}]$. When the spacing is either too small (close to $s_{st}$) or too large (close to $s_{go}$), drivers can easily follow the desired optimal velocity by driving very slowly ($v^*$ is near $0$) or following the maximum speed ($v^*$ is near $v_{max}$). The resulting uncontrolled system hence tends to be more stable. However, when the spacing is close to $\frac{s_{go} - s_{st}}{2}$ (the red star in Fig.~\ref{fig:interp}), the original system becomes more unstable, since slight changes in spacing leads to large variations in the desired optimal velocity. In fact, the default $s_{st} = 5, s_{go} = 35$ place the default spacing $L / n = 20 m$ at the most unstable inflection point (the red star). Similar interpretations can be applied to the positioning of two boundary values $(s_{st}, s_{go})$. Notably, the hold limit variation in $(L, n, s_{st}, s_{go})$ is mild, ranging from $1.5s$ to $2s$ in simulation, as these four variables are all encapsulated within a cosine function of the desired optimal velocity.

\begin{figure}
    \centering
    \includegraphics[width=0.35\textwidth]{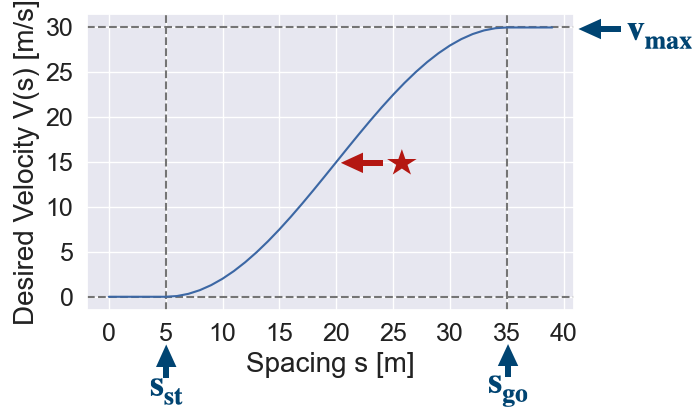}
    \caption{The Optimal Velocity function $V(s)$ in Eq.~\eqref{eq:vclipping} and~\ref{eq:ovm_f} with default parameters in Sec.~\ref{sec:experiment_setup}. The red star represents the equilibrium spacing and velocity with the default parameters, where the function attains maximum slope. Changing system parameters moves the red star to different positions on the curve, affecting the stability of the uncontrolled system and the hold limit to stabilize the system. \vspace{-0.2cm}}
    \label{fig:interp}
\end{figure}

In contrast, the maximum desired velocity (speed limit) $v_{max}$ acts as a multiplier for the desired optimal velocit and has a more substantial impact on the hold limit: as $v_{max}$ increases from $25m/s$ to $60m/s$, the hold limit decreases from $2s$ to $0s$. Increasing $v_{max}$ stretches the desired velocity curve, resulting in sharper changes of the desired optimal velocity in response to spacing variations. Consequently, a higher $v_{max}$ yields a more unstable uncontrolled system that leads to a shorter hold limit. 
While the stability of the uncontrolled system explains a linear decrease in the hold limit, we observe a super-linear decrease in the simulation due to two additional factors: (1) the larger magnitude of the controller, as shown in Fig.~\ref{fig:interp_subp}, introduces more errors to the system through the piecewise-constant hold, and (2) the unstable system leads to vehicle collisions, further complicating the task of stabilizing the system with a noisy controller.

\textbf{Low sensitivity to relative position, high sensitivity to relative speed}: The remaining parameters, $\alpha$ and $\beta$, indicate the sensitivity of human drivers to the desired optimal velocity ($\alpha$) and the velocity of the preceding vehicle ($\beta$) in comparison to the ego velocity. Interestingly, we observe different trends of the simulation hold limits for the two parameters, although larger $\alpha$ and $\beta$ both results in increased stability in the original uncontrolled system (Eq.~\eqref{eq:string_ovm}). For $\beta$, the inclusion of the velocity dissipation term enhances the driver's awareness of their surroundings, leading to improved system stability. The sharp, super-linear decrease in the hold limit for small $\beta$ values arises from similar factors as those affecting $v_{max}$, which combines (1) uncontrolled system's stability, (2) additional errors induced due to the controller's large magnitude, and (3) vehicle collisions when the system is excessively unstable. 

In contrast, for $\alpha$, the simulation hold limit displays an opposite trend to the stability of the uncontrolled system, albeit with relatively mild variation ($1.6s$ to $2s$). As observed in Fig.~\ref{fig:interp_subp}, large $\alpha$ results in more stable uncontrolled systems, but also larger controller magnitude, and hence holding the control piecewise-constant adds more noise to the system. This can be explained by the fact that a larger $\alpha$ corresponds to human drivers adhering more strongly to the suggested optimal velocity, resulting in a more stable uncontrolled OVM system. However, the optimal velocity prescribed by the OVM may conflict with the actions of the controlled vehicle. With both longer hold length and larger $\alpha$, the controlled vehicle may open up wider gaps, resulting in a stronger response from the following human drivers, in turn causing system instabilities. 

\textbf{Insights for traffic system design: } Based on the above interpretations, transporation designers can select system parameters to enable effective deployment of coarse-grained guidance, for example, by (1) adjusting speed limits or (2) ensuring roads provide clear visibility of the traffic or equipping vehicles with sensors for adaptive cruise control to enhance human driver's awareness to preceding traffic (higher $\beta$).
\mbox{}

\vspace{\squeezeSection}
\vspace*{0.2cm}
\subsection{Controller design for coarse-grained guided driving}

Thus far, we have focused on analyzing a \textit{given} controller, the continuous $\mathcal{H}_2$ optimal controller with a simulation hold limit of $1.66s$ by default.
In this section, we consider several approaches to intentionally design controllers for coarse-grained guidance to achieve system-level traffic flow stability. We keep the OVM system parameters $(L, n, s_{st}, s_{go}, v_{max}, \alpha, \beta)$ at the default values in Sec.~\ref{sec:experiment_setup}.

\textbf{Lyapunov-Krasovskii controller search:} Recall that the Lyapunov-Krasovskii analysis provides a method to obtain piecewise-constant controllers. Here, we examine the quality of the controllers via simulation: we first solve LMIs~(\ref{eq:Lyapunov_Krasovskii_control}) for a control gain matrix $K_{LK} = L Q^{-1}$, with a grid search of input hold length parameter $\Delta_{in} \in \{1, 2, 3, 4, 5, 6, 7, 8\}$. We fix the tuning parameter $\epsilon = 1$ in LMI (\ref{eq:Lyapunov_Krasovskii_control}) where we substitute $P_3 = \epsilon P_2$ from (\ref{eq:Lyapunov_Krasovskii}), as we empirically find such an $\epsilon$ yields the best controller with the longest simulation hold limit. Given the resulting control gain matrix $K_{LK}$, we then perform simulation via a binary search with a granularity of $T_{step}=0.01s$ to examine the empirical hold limit $\Delta_{sim}$.

Table~\ref{tab:LK_sim} displays the simulation hold limits $\Delta_{sim}s$ for different input parameters $\Delta_{in}s$. We observe that, as $\Delta_{in}$ increases, the Lyapunov-Krasovskii analysis finds better controllers with longer hold limits, reaching a hold limit of $4.55s$ at $\Delta_{in} = 3s$ that is a 2.7x improvement from the continuous $\mathcal{H}_2$ optimal controller. Fig.~\ref{fig:controller_vis} visualizes the controller profiles for the default continuous $\mathcal{H}_2$ optimal controller and the Lyapunov-Krasovskii controller with $\Delta_{in} = 3s$ by plotting the controllers' gain matrices ($K_{default}$ and $K_{LK}$) associated with the spacing and velocity ($\tilde{s}_i = s_i - s^*$ and $\tilde{v}_i = v_i - v^*$ $\forall i$). We observe that, for spacing, the Lyapunov-Krasovskii controller considers more vehicles ahead of the guided vehicle ($\tilde{s}_1$) than the continuous $\mathcal{H}_2$ optimal controller, placing the highest weight on the $5^{th}$ vehicle ahead ($\tilde{s}_5$). For velocity, while both controllers heavily focus on the ego guided vehicle ($\tilde{v}_1$), the Lyapunov-Krasovskii controller also takes into account a few vehicles located behind (e.g. positive weight on $\tilde{v}_{20}$ and $\tilde{v}_{19}$).

\begin{figure}[!tb]
     \centering
    \includegraphics[width=0.24\textwidth]{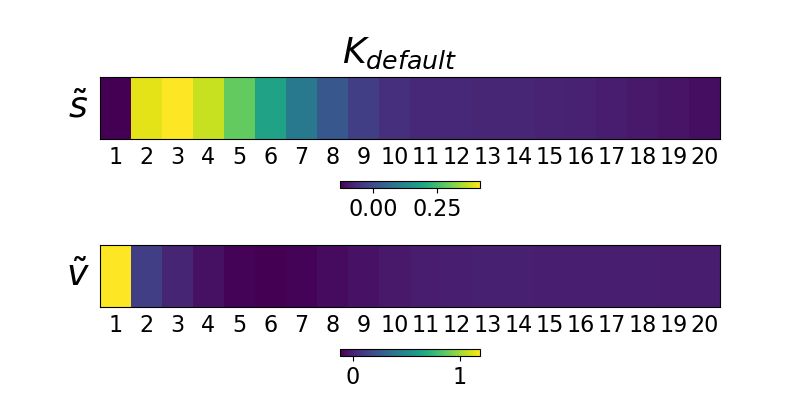}
    \includegraphics[width=0.24\textwidth]{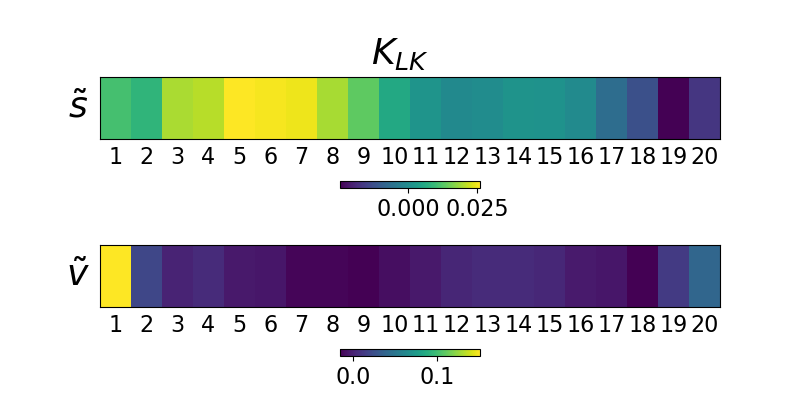}
    \caption{Controller profile visualization: (Left) the default continuous $\mathcal{H}_2$ optimal controller gain matrix $K_{default}$; (Right) the Lyapunov-Krasovskii controller gain matrix $K_{LK}$ with $\Delta_{in} = 3s$; (Up) the controllers' weights on the headway error states $\tilde{s}_i$; (Down) the controllers' weights on the velocity error states $\tilde{v}_i$. \vspace{-0.6cm}}
    \label{fig:controller_vis}
\end{figure}

\begin{table}[!ht]
\centering
\resizebox{.49\textwidth}{!}{
\begin{tabular}{ |c|c|c|c|c|c|c|c|c|c|c|c|} 
 \hline
 $\Delta\mathbf{_{in} (s)}$ & $1$ & $2$ & $3$ & $4$ & $5$ & $6$ & $7$ & $8$\\
 \hline
 $\Delta\mathbf{_{sim} (s)}$ &  $2.71$ & $3.33$ & $\mathbf{4.55}$ & $4.54$ & $4.3$ & $4.19$ & $4.09$ & $3.87$\\ 
  \hline
  \text{improv.} & 1.6x & 2.0x & 2.7x & 2.7x & 2.6x & 2.5x & 2.5x & 2.3x \\
  \hline
\end{tabular}
}
 \caption{The simulation hold limit $\Delta_{sim}$ with the Lyapunov-Krasovskii control gain matrix when we set $\Delta_{in} \in \{1, 2, ..., 8\}$ in LMIs~(\ref{eq:Lyapunov_Krasovskii_control}). The last row displays the relative improvement in $\Delta_{sim}$ from the continuous $\mathcal{H}_2$ optimal controller.}
 \label{tab:LK_sim}
\end{table}

When the input $\Delta_{in}$ further increases over $4s$, the simulation hold limit $\Delta_{sim}$ decreases, again due to collisions in the system not modeled by the theoretical analysis (as discussed in Sec.~\ref{sec:setup_stability_criteria}). Notably, same as in Sec.~\ref{sec:experiment2}, if we omit the collision constraint in simulation, we would achieve increasing simulation hold limits $\Delta_{sim}s$ with increasing large $\Delta_{in}s$, which confirms the ability of Lyapunov-Krasovskii analysis to stabilize the system if we ignore additional constraints.

\textbf{$\mathcal{H}_2$ re-scaling}: Next, we propose and examine a heuristic controller design where we vary the control parameters $(k_{mult}, \gamma_s, \gamma_v)$ to find scaled controllers more suitable for the sample-data system. In Fig.~\ref{fig:sensitivity}, we observe that controllers with smaller magnitudes than the default continuous $\mathcal{H}_2$ controllers, given by smaller $k_{mult} < 1$, $\gamma_s < 0.03$, $\gamma_v < 0.15$, result in longer hold limits. Notably, the longest hold limit of $4.78s$ is achieved when $k_{mult} = 0.2$, offering a 2.9x improvement from the default when $k_{mult} = 1$. This can be explained by the $\sigma_{max}(A_1)$ term in the denominator of the Lyapunov analysis in Eq.~\eqref{eq:lyapunov_theory}, where controllers of larger magnitudes incur larger errors from the piecewise-constant hold. However, when the controller is too small, it is not powerful enough to stabilize the system, resulting in a decrease in hold limit (e.g. from $4.78s$ when $k_{mult} = 0.2$ to $2.84s$ when $k_{mult} = 0.005$). This can be explained by the ratio $\sigma_{min}(Q) / \sigma_{max}(P)$ in the Lyapunov analysis, which represents the stability of the controlled system $A - k_{mult} BK$. Hence, there is a trade-off between the controller's power and the noise incurred from the piecewise-constant hold. 

Meanwhile, as the best simulation hold limit of the best scaled continuous $\mathcal{H}_2$ optimal controllers is around the same level as the Lyapunov-Krasovskii controllers ($4.55 - 4.78s$), we observe that the piecewise-constant controller obtained by scaling down a reasonable continuous control may perform well for coarse-grained guidance. Given abundant learning-based controllers developed for the continuous traffic systems~\cite{wu2021flow, yan2022unified}, promising strategies for coarse-grained guidance include taking the down-scaled versions of the existing controllers or finetuning the controllers with a penalization on the magnitude to avoid retraining the controllers from scratch.

\section{Extensions to Human Errors}
\textbf{Non-vanishing magnitude error:} We simulate our system with magnitude errors in Eq.~\eqref{eq:magnitude_error_system} under the default OVM parameter and the default continuous $\mathcal{H}_2$ optimal controller. To simulate the system under nonvanishing magnitude errors, we add $\sigma_i\cdot d_{nv}$ to the acceleration of each vehicle (the guided HV and $n-1$ unguided HVs following OVM), where $\sigma_i \stackrel{i.i.d.}{\sim} Bernoulli(0.5)$ and we vary $d_{nv}$ across a range of values. We confirm from the simulation that the system converges to a bounded region around equilibrium (the ultimate bound) under the same hold limit $\Delta=1.66s$ as the original system in Eq.~\eqref{eq:matrix_system}.
Fig.~\ref{fig:human_error} (Top) presents the simulation ultimate bound as a function of error magnitude $d_{nv}$. The ultimate bound is computed as the largest $\|x\|_2$, where $x$ consists of the headway and velocity error states of all vehicles, across the last 10\% of all $50$ simulation trajectories. We verify that the ultimate bound converges by comparing it with the largest $\|x\|_2$ across all states at the last 20\% and 30\% of the simulation trajectories. We observe a linear increase in the simulation ultimate bound as the error magnitude increases, which aligns with our theoretical extension in Prop.~\ref{prop:magnitude_error_lyapunov}. 

\begin{figure}[!thb]
     \centering
    \includegraphics[width=0.2\textwidth]{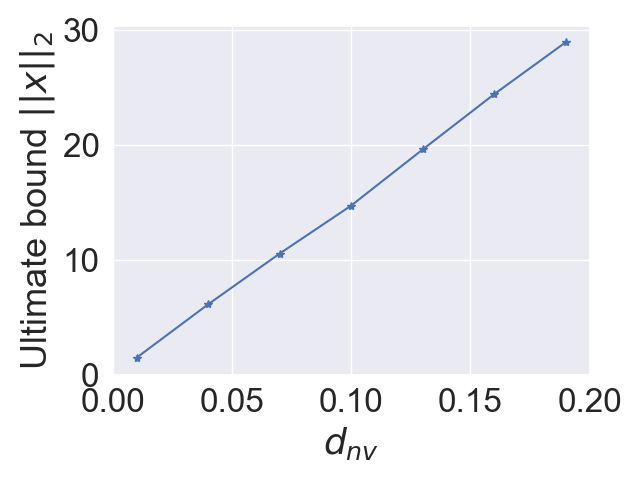}\\
    \includegraphics[width=0.2\textwidth]{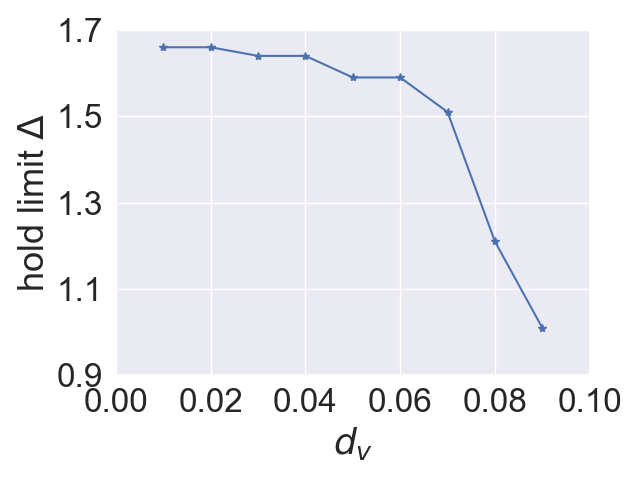}
    \includegraphics[width=0.2\textwidth]{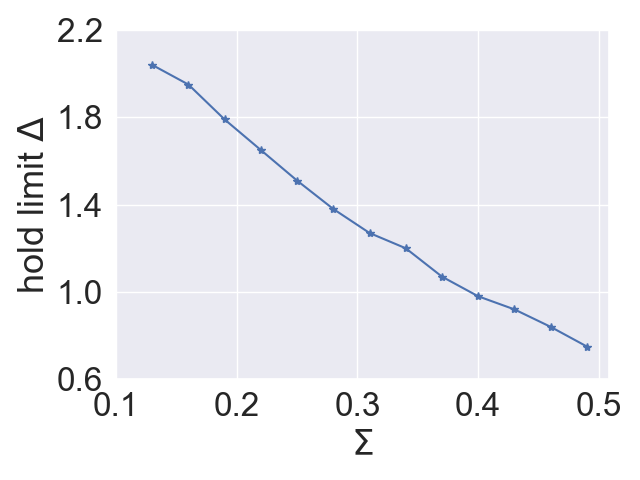}
    \caption{Top: \textit{nonvanishing magnitude error}. We plot the ultimate bound $\|x\|_2$ for a variety of magnitude $d_{nv}$. Bottom Left: \textit{vanishing magnitude error}. We plot the hold limit $\Delta$ for a variety of magnitude $d_{v}$. Bottom right: \textit{reaction delay}. We plot the hold limit $\Delta$ for a variety of human reaction delay $\Sigma$.}
    \label{fig:human_error}
\end{figure}

\textbf{Vanishing magnitude error:} Similarly, we simulate the system under vanishing magnitude errors by adding $\sigma_i\cdot d_{v} \cdot \|x(t)\|_2$ to the acceleration of each vehicle at simulation step $t$. $\sigma_i$ remains the same as mentioned earlier, and we vary $d_v$ across a range of values. As shown in Fig.~\ref{fig:human_error} (Bottom left), the system converges to equilibrium with a smaller simulation hold limit that exhibits a linear decrease within $d_v \in [0, 0.07]$. Our theoretical extension in Prop.~\ref{prop:lyapunov} confirms the linear decrease near equilibrium. Beyond $d_v = 0.07$, the hold limit decreases at a faster rate. As the vanishing perturbation scales proportionally with the norm of the state, a large magnitude error (large $d_v$) pushes the system away from the equilibrium, leading to a compounding effect that amplifies the error magnitude and moves the system even farther away. While the linearization of our theoretical analysis provides accurate guarantees near equilibrium, its accuracy diminishes as the system moves further away. We leave as a future work to enhance the theoretical analysis for the case of large $d_v$.

\textbf{Reaction delay:} Finally, we simulate our system with human reaction delay in Eq.~\eqref{eq:reaction_delay_system} under the default OVM parameter and the default continuous $\mathcal{H}_2$ optimal controller across a range of reaction delay values $\Sigma$. In Fig.~\ref{fig:human_error} (Bottom right), we observe the simulation hold limit $\Delta$ decreases linearly as the human reaction delay $\Sigma$ increases. The simulation result validates our theoretical extension in Prop.~\ref{prop:reaction_delay_lyapunov} that establishes a linear relationship between the increase in reaction delay and the decrease in the hold limit.
\color{black}

\section{Conclusion}
This work presents an integrated Lyapunov analysis framework of coarse-grained guidance, a class of policies that aim to guide human drivers to stabilize the traffic to bypass the difficulty of AV deployment. We derive both a Lyapunov analysis for qualitative interpretation of the relationships between traffic system parameters and the hold limit, and a Lyapunov-Krasovskii analysis for quantitative estimation of the hold limit and for controller design. Our work highlights the Lyapunov analysis framework as an important integrated theoretical tool for obtaining efficient, safe, and sustainable transportation systems under coarse-grained guidance.

We propose a few important directions for future research. First, we would like to tighten the derivation of the Lyapunov analysis (Eq.~\eqref{eq:lyapunov_theory}) to obtain absolute scales of different components in the bound. The correct scaling will enable us to pinpoint the exact slope and location of the optimum of the curves in Fig.~\ref{fig:sensitivity}, while the current bound is only able to describe the relative trend. Next, we would like to incorporate control barrier functions to the Lyapunov-Krasovskii analysis (LMIs~(\ref{eq:Lyapunov_Krasovskii}) and (\ref{eq:Lyapunov_Krasovskii})) to tighten the bound under unsafe events such as collision. Finally, we would like to consider expanding our theory to a broader class of human-compatible driving policies that consist of other easy-to-follow driving instructions, as well as to more complex traffic scenarios. 

\section*{ACKNOWLEDGMENT}
This work was supported by the National Science Foundation (NSF) under grant number 2149548, the MIT Amazon Science Hub, the MIT Energy Initiative (MITEI) Mobility Systems Center, MIT’s Research Support Committee, as well as a gift from Mathworks.

\vspace{-0.2cm}
\bibliography{references}

\vspace{\squeezeSection}
\vspace{\squeezeSection}

\begin{IEEEbiography}[{\includegraphics[width=1in,height=1.25in,clip,keepaspectratio]{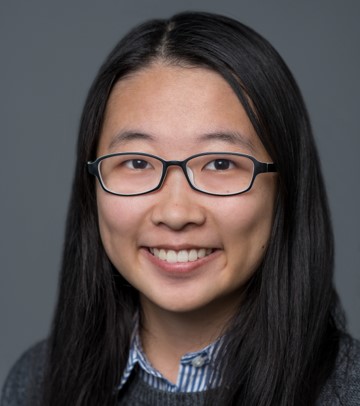}}]{Sirui Li} received the B.S. degree with majors in computer science and mathematics from Washington University in St. Louis in 2019. She is currently working toward the Ph.D. degree in Social and Engineering System at Massachusetts Institute of Technology, Cambridge.
Her research interests include areas of machine learning for combinatorial optimization and control analysis for transportation systems.
\end{IEEEbiography}
\vspace{-1.6 cm}

\begin{IEEEbiography}[{\includegraphics[width=1in,height=1.25in,clip,keepaspectratio]{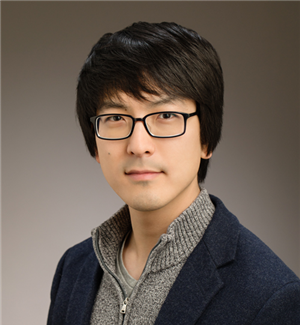}}]{Roy Dong} is an Assistant Professor in the Industrial \& Enterprise Engineering department at the University of Illinois at Urbana-Champaign. He received a BS (Hons.) in Computer Engineering and a BS (Hons.) in Economics from Michigan State University in 2010 and the PhD in Electrical Engineering and Computer Sciences from UC Berkeley in 2017. Prior to his current position, he was a postdoctoral researcher in the Berkeley Energy \& Climate Institute, a visiting lecturer in the Industrial Engineering and Operations Research department at UC Berkeley, and a Research Assistant Professor in the Electrical and Computer Engineering department at the University of Illinois at Urbana-Champaign. 
\end{IEEEbiography}
\vspace{-1.4 cm}
\begin{IEEEbiography}[{\includegraphics[width=1in,height=1.25in,clip,keepaspectratio]{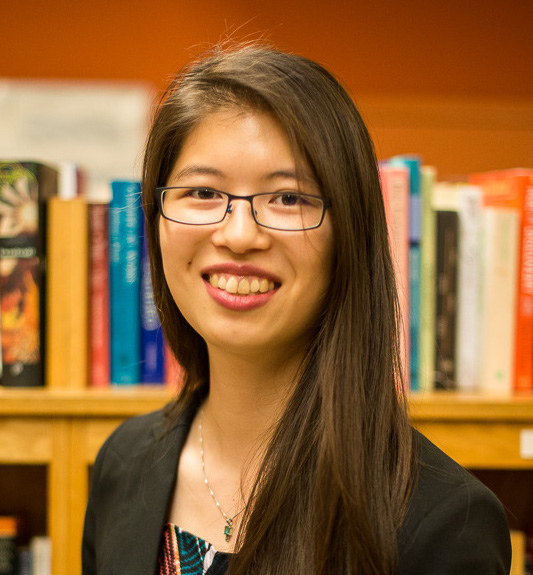}}]{Cathy Wu} Cathy Wu is an Assistant Professor at MIT in LIDS, CEE, and IDSS. She holds a Ph.D. from UC Berkeley, and B.S. and M.Eng. from MIT, all in EECS, and completed a Postdoc at Microsoft Research. Her research interests are at the intersection of machine learning, control, and mobility. Her recent research focuses on how learning-enabled methods can better cope with the complexity, diversity, and scale of control and operations in mobility systems. She is broadly interested in developing principled computational tools to enable reliable decision-making in sociotechnical systems.
\end{IEEEbiography}

\newpage
\section*{APPENDIX}
\subsection{Proofs for extensions to human errors}
\subsubsection{Proof of Proposition~\ref{prop:magnitude_error_lyapunov}}\label{proof:magnitude_error_lyapunov}
\begin{proof} 
We can follow the same Lyapunov Analysis in Sec.~\ref{sec:lyapunov1} with the modified dynamics~(\ref{eq:magnitude_error_system}), and arrive at the following equation
\begin{equation}
\begin{aligned}
& \quad \left(V(x(t)) - V(x(t_k))\right)_{new} \\
& = \left(V(x(t)) - V(x(t_k))\right)_{old} \\
& + \langle 2x(t_k)^\intercal P,\; B_d d(t_k)\rangle + \langle 2x(t_k)^\intercal P, \; B_d (d(t^*) - d(t_k)) \rangle\\
& + \langle 2(x(t^*) - x(t_k))^\intercal P,\; B_d d(t_k)\rangle) (t - t_k). \label{eq:appendix_magnitude_error_lyapunov_block1}
\end{aligned}
\end{equation}
where $\left(V(x(t)) - V(x(t_k))\right)_{old}$ is the same as Eq.~\eqref{eq:lyapunov_theory_block1} and can be bounded by Eq.~\eqref{eq:lyapunov_theory_block5}. We can bound the additional human error terms (the last three terms), denoted as $P_m$, as
\begin{equation}
\begin{aligned}
    P_m & \leq \sigma_{max}(B_dP + P B^\intercal_d) \big(\|x(t_k)\|_2 \|d(t_k)\|_2 \\
    & + \|x(t_k)\|_2\|d(t^*) - d(t_k)\|_2 + \|x(t^*) - x(t_k)\|_2\|d(t_k)\|_2\big).\label{eq:d_perturb}
\end{aligned}
\end{equation}
Substituting $\|x(t^*) - x(t_k)\|_2 \leq \|x(t^*)\|_2 + \|x(t_k)\|_2, \|d(t^*) - d(t_k)\|_2 \leq \|d(t^*)\|_2 + \|d(t_k)\|_2$, we obtain an upper bound on Eq.~\eqref{eq:d_perturb} as follows:
\begin{equation}
\begin{aligned}
    P_m & \leq \sigma_{max}(B_dP + P B^\intercal_d) \big(3\|x(t_k)\|_2 \|d(t_k)\|_2 \\
    & + \|x(t_k)\|_2\|d(t^*)\|_2 + \|x(t^*)\|_2\|d(t_k)\|_2 \big).
\end{aligned}
\end{equation}
\textit{Nonvanishing Perturbation:} We have 
\begin{equation}
\begin{aligned}
    P_m & \leq 5 \sigma_{max}(B_dP + P B^\intercal_d) \bar{d}_{nv} \max\limits_{s \in [t_k, t_{k+1}]} \left\|x(s)\right\|_2.
\end{aligned}
\end{equation}
Plugging the above into Eq.~\eqref{eq:appendix_magnitude_error_lyapunov_block1}, we get 
\begin{equation}
\begin{aligned}
    & \quad \left(V(x(t)) - V(x(t_k))\right)_{new} \\
    & \leq (t-t_k)\Big(-N \|x(t_k)\|_2 + M \Big)\max\limits_{s \in [t_k, t_{k+1}]} \|x(s)\|_2. \label{eq:appendix_magnitude_error_lyapunov_block2}
\end{aligned}
\end{equation}
where we denote
\begin{equation}
    \begin{aligned}
        N & := \sigma_{\min}(Q) - c \Delta \cdot \sigma_{\max}(P) (\sigma_{\max}(A) + \sigma_{\max}(A_1))^2, \\
        M & := 5 \sigma_{max}(B_dP + P B^\intercal_d) \bar{d}_{nv}.
    \end{aligned} \label{eq:appendix_magnitude_error_lyapunov_block3}
\end{equation}
We hence arrive at an ultimate bound~\cite{khalil2001nonlinear} as follows
\begin{equation}
    \begin{aligned}
        - N \|x(t_k)\|_2 + M &= -(1-\theta) N \|x(t_k)\|_2 - \theta N \|x(t_k)\|_2 + M, \\ & \quad \quad \quad \forall \; 0 < \theta < 1\\
    \Rightarrow - N \|x(t_k)\|_2 + M & \leq -(1-\theta)N\|x(t_k)\|_2, \quad \forall \; \|x(t_k)\|_2 \geq \frac{M}{\theta N}\\
    \Rightarrow - N \|x(t_k)\|_2 + M & \leq -\frac{1}{2}N\|x(t_k)\|_2, \quad \forall \; \|x(t_k)\|_2 \geq \frac{2M}{N} \\
        & \quad \quad \quad \text{ by setting } \theta = \frac{1}{2}.
    \end{aligned}
    \label{eq:ultimate_bound_deriv1}
\end{equation}
From the original proof we know for some $d > 1$, 
\begin{equation}
    \begin{aligned}
        \Delta \; & \leq c' \frac{\sigma_{\min}(Q)}{\sigma_{\max}(P) (\sigma_{\max}(A) + \sigma_{\max}(A_1))^2}  \Leftrightarrow \;\; -N 
        \leq -\frac{1}{d}\sigma_{\min}(Q). \\
    \end{aligned}
\end{equation}
By Eq.~\eqref{eq:ultimate_bound_deriv1}, the above bound on the hold limit $\Delta$ results in the following decrease in the Lyapunov value:
\begin{equation}
    \left(V(x(t)) - V(x(t_k))\right)_{new} \leq -\frac{1}{2d}\sigma_{min}(Q)(t-t_k)\max\limits_{s \in [t_k, t_{k+1}]}\|x(s)\|_2, 
\end{equation}
if the state is outside the following bounded region 
\begin{equation}
    \|x(t)\|_2 \geq \frac{10 \sigma_{max}(B_dP + P B^\intercal_d)}{1/d\sigma_{\min}(Q)}\bar{d}_{nv}.
\end{equation}
Combining the above, we conclude the convergence of the trajectory within the bounded region around equilibrium (the ultimate bound) in Prop.~\ref{prop:magnitude_error_lyapunov} under vanishing error, when the hold limit is the same as in the original system~\eqref{eq:matrix_system}.\\
\textit{Vanishing Perturbation:} We have 
\begin{equation}
    P_m \leq 5 \sigma_{max}(B_dP + P B^\intercal_d) \bar{d}_v\|x(t_k)\|_2 \max\limits_{s \in [t_k, t_{k+1}]}\|x(s)\|_2.
\end{equation}
Plugging the above into Eq.~\eqref{eq:appendix_magnitude_error_lyapunov_block1}, we get 
\begin{equation}
    \begin{aligned}
        & \quad \left(V(x(t)) - V(x(t_k))\right)_{new} \\
        & \leq (t-t_k)\big(-N + M) \|x(t_k)\|_2 \max\limits_{s \in [t_k, t_{k+1}]} \|x(s)\|_2.
    \end{aligned}
\end{equation}
With the same N and M as in Eq.~\eqref{eq:appendix_magnitude_error_lyapunov_block3}, with $\bar{d}_{nv}$ replaced by $\bar{d}_{v}$. Following a similar proof as in Sec.~\ref{sec:lyapunov1}, in order for $\left(V(x(t)) - V(x(t_k))\right)_{new}$ to have a sufficient decrease for the system to converge to equilibrium, e.g. for some $d > 1$, 
\begin{equation}
    \resizebox{.49\textwidth}{!}{
    $\left(V(x(t)) - V(x(t_k))\right)_{new} \leq -(t-t_k) \frac{\sigma_{\min}(Q)}{d} \|x(t_k)\|_2 \max\limits_{s \in [t_k, t_{k+1}]} \|x(s)\|_2,$
    }
\end{equation}
the following establishes a sufficient condition that gives a slightly smaller $\Delta$ due to the vanishing magnitude error (reflected in the additional term in the numerator).
\begin{equation}
     - N + M \leq \frac{\sigma_{\min}(Q)}{d} \Leftrightarrow \;\; \Delta \; \leq c' \frac{\sigma_{\min}(Q) - 5 \sigma_{max}(B_dP + P B^\intercal_d)\bar{d}_v}{\sigma_{\max}(P) (\sigma_{\max}(A) + \sigma_{\max}(A_1))^2}, 
\end{equation}
for the same $c' > 0$.
\end{proof}

\subsubsection{Proof of Proposition~\ref{prop:magnitude_error_lk}}
\begin{proof} \label{proof:magnitude_error_lk} The Lyapunov-Krasovskii functional for our sample-data system with $P > 0, U > 0$ is
\begin{equation}
    \begin{aligned}
         V(t, x(t), \dot{x}(t)) = x^\intercal(t) P x(t) + (\Delta - \tau(t)) \int_{t - \tau(t)}^{t}\dot{x}^\intercal(s) U \dot{x}(s) ds,
    \end{aligned}
\end{equation}
as $\dot{\tau}(t) = 1$ and $\frac{d}{dt}[t - \tau(t)] = 1 - \dot{\tau}(t) = 0$, the above satisfies
\begin{equation}
    \begin{aligned}
    \frac{\partial}{\partial t}V(t, x(t), \dot{x}(t)) &  = 2\dot{x}^\intercal(t) P x(t) - \int_{t - \tau(t)}^{t}\dot{x}^\intercal(s) U \dot{x}(s) ds \\
    & + (\Delta - \tau(t))\dot{x}^\intercal(t) U \dot{x}(t).
    \label{eq:magnitude_error_lk_deriv1}
    \end{aligned}
\end{equation}
Following the original proof, we let $v_1 = \frac{1}{\tau(t)}\int_{t - \tau(t)}^{t} \dot{x}(s) ds$, so $\int_{t - \tau(t)}^{t}\dot{x}^\intercal(s) U \dot{x}(s) ds \geq \tau(t) v_1^\intercal U v_1$ by Jensen's inequality. 
As our dynamics under magnitude error~(\ref{eq:magnitude_error_system}) can be written as
\begin{equation}
    \dot{x}(t) = (A + A_1)x(t) - \tau(t) A_1 v_1 + B_d d(t).
\end{equation}
We use the descriptor method to bound the positive terms (first and third) in Eq.~\eqref{eq:magnitude_error_lk_deriv1} with
\begin{equation}    
    \resizebox{.49\textwidth}{!}{
        $0 = 2[x(t)^\intercal P_2^\intercal + \dot{x}(t)^\intercal P_3^\intercal ][(A + A_1)x(t) - \tau(t)A_1 v_1 + B_d d(t) - \dot{x}(t)]$.
    }
\end{equation}
where $P_2, P_3$ are slack matrices. We can write Eq.~\eqref{eq:magnitude_error_lk_deriv1} as 
\begin{equation}
    \begin{aligned}
        \frac{\partial}{\partial t}V(t, x(t), \dot{x}(t)) & \leq \eta_1(s)^\intercal(t) \Psi_s \eta_1(s) \\
        & + x(t)^\intercal P_2^\intercal B_d d(t) + \dot{x}(t)^\intercal P_3^\intercal B_d d(t) \\
        & = \eta_2^\intercal(t) \tilde{\Psi}_s \eta_2(s) < -\epsilon |x(t)|^2 \text{ for some } \epsilon > 0, 
    \end{aligned}
\end{equation}
where $\eta_1 = col\{x(t), \dot{x}(t), v_1\}$, $\eta_2 = col\{x(t), \dot{x}(t), v_1, d(t)\}$, $\Psi_s$ is the upper left $3\times 3$ block of $\tilde{\Psi}_s$, and
\begin{equation}
\begin{aligned}
    \tilde{\Psi}_s = \begin{bmatrix}
    \Phi_{11} & P - P_2^\intercal + (A + A_1)^\intercal P_3 & -\tau(t) P_2^\intercal A_1 & P_3^\intercal B_d\\
    * & -P_3 - P_3^\intercal + (\Delta - \tau(t))U & -\tau(t) P_3^\intercal A_1 & P_3^\intercal B_d\\
    * & * & -\tau(t) U & 0\\
    * & * & * & 0
    \end{bmatrix}.
\end{aligned}
\end{equation}
The system is then $H_{\infty}$ robust at disturbance attenuation level $\gamma > 0$ if the following $LMI$s holds
\begin{equation}
\begin{aligned}
    \begin{bmatrix}
    \Phi_{11} + I& P - P_2^\intercal + (A + A_1)^\intercal P_3 & -\tau(t) P_2^\intercal A_1 & P_2^\intercal B_d\\
    * & -P_3 - P_3^\intercal + (\Delta - \tau(t)) U & -\tau(t) P_3^\intercal A_1 & P_3^\intercal B_d\\
    * & * & -\tau(t) U & 0 \\
    * & * & * & -\gamma^2 I
    \end{bmatrix}  < 0, 
\end{aligned}
\end{equation}
as we would have $\frac{\partial}{\partial t} V(t) + x(t)^\intercal x(t) - \gamma^2 d(t)^\intercal d(t) < 0$, with $lim_{t \rightarrow \infty} V(t) = 0$, under the zero initial condition. Integrating both sides of the above from $0$ to $+\infty$, we recover the $H_{\infty}$ robustness in Definition~\ref{def:Hinf_robust}. As $0 \leq \tau(t) \leq \Delta$, we let $\tau(t)\rightarrow 0$ and $\tau(t) \rightarrow \Delta$ and get Eq.~\eqref{eq:magnitude_error_Lyapunov_Krasovskii}.
\end{proof}

\subsubsection{Proof of Proposition~\ref{prop:reaction_delay_lyapunov}}\label{proof:reaction_delay_lyapunov}
\begin{proof} 
Again, we can follow the same Lyapunov Analysis in Sec.~\ref{sec:lyapunov1} with the modified dynamics~(\ref{eq:reaction_delay_system}), and arrive at the following equation
\begin{equation}
\begin{aligned}
& \quad \left(V(x(t)) - V(x(t_k))\right)_{new} \\
& = \left(V(x(t)) - V(x(t_{k'}))\right)_{old} + \langle 2x(t^*)^\intercal P,\; A_1 (x(t_k) - x(t_{k'}))\rangle .\label{eq:reaction_delay_deriv1}
\end{aligned}
\end{equation}
The first term can be bounded using the same analysis as in the original proof as
\begin{equation}
\begin{aligned}
\leq  &  \left(-\sigma_{min}(Q) + c \Delta \cdot \sigma_{\max}(P) (\sigma_{\max}(A) + \sigma_{\max}(A_1))^2 \right)\\
& \|x(t_{k'})\|_2 \max\limits_{s \in [t_{k'}, t_{k+1'}]} \|x(s)\|_2,  \text{ for some constant } c.
\end{aligned}
\end{equation}
The second term, denoted as $P_r$, is the additional noise terms from human reaction delay. We have
\begin{equation}
\begin{aligned}
    & \left\|x(t_{k'}) - x(t_k)\right\|_2 = \left\|\int_{t_k}^{t_{k'}} \dot{x}(s) ds\right\|_2 \leq (t_{k'} - t_k) \max\limits_{s \in [t_{k}, t_{k'}]} \|\dot{x} (s)\|_2.
\end{aligned}
\end{equation}
Notably, within the internal $[t_k, t_k')$, the guided HV follows the instruction from the previous holding period $[t_{k-1'}, t_k')$. Under the following assumption (where the constant $\bar{D}_v \geq 0$ indicates how smooth the instructions are updated):
\begin{equation}
    \begin{aligned}
    \|\dot{x} (s)\|_2 & \leq \bar{D}_{v} \max\limits_{s \in [t_{k'}, t_{k'+1}]}\|\dot{x}(t_{s})\|_2 \quad \text{ for some constant } \bar{D}_v\\
    & \leq \bar{D}_{v} \sigma_{max}(A +A_1)\max\limits_{s \in [t_{k'}, t_{k'+1}]}\|x(t_{s})\|_2 \quad \forall s \in [t_{k}, t_{k'}],
    \end{aligned}
\end{equation}
the second term can be bounded by
\begin{equation}
    \begin{aligned}
        P_r & \leq \sigma_{max}(A_1 P + P A_1^\intercal) \Sigma \bar{D}_{v} \sigma_{max}(A +A_1) \max\limits_{s \in [t_{k'}, t_{k'+1}]}\|x(t_{s})\|_2 \\
        & \leq \sigma_{max}(P)(\sigma_{max}(A) + \sigma_{max}(A_1))^2\Sigma \bar{D}_v \max\limits_{s \in [t_{k'}, t_{k'+1}]}\|x(t_{s})\|_2
        \label{eq:bound_on_deriv}
    \end{aligned}
\end{equation}
Hence, we follow the same argument as in Prop.~\ref{prop:magnitude_error_lyapunov} (vanishing perturbation) and conclude a slightly smaller $\Delta$ for the system to converge to equilibrium (due to reaction delay, as reflected in the additional term in the numerator).
\begin{equation}
    \begin{aligned}
        & \sigma_{\max}(P) (\sigma_{\max}(A) + \sigma_{\max}(A_1))^2 \left(c \Delta  + \Sigma \bar{D}_{v} \right) \leq \frac{d-1}{d}\sigma_{\min}(Q) \\
        \Leftrightarrow \;\; & \Delta \; \leq c' \frac{\sigma_{\min}(Q) }{\sigma_{\max}(P) (\sigma_{\max}(A) + \sigma_{\max}(A_1))^2} - c'' \Sigma \bar{D}_v, 
    \end{aligned}
\end{equation}
for some $c' > 0, c'' > 0$. 
\end{proof}

\end{document}